\documentclass[12pt,a4paper]{article}
\usepackage[colorlinks=true,linkcolor=blue,filecolor=blue,urlcolor=blue,citecolor=blue]{hyperref}
\usepackage{amsfonts,amssymb,amsmath,amsthm,latexsym,a4wide}
\usepackage{rotating}
\usepackage{lscape}
\usepackage{graphics}
\usepackage{mathrsfs,color}
\usepackage[round]{natbib}
\usepackage{booktabs}
\usepackage{setspace}

\newtheorem{assumption}{Assumption}
\newtheorem{lemma}{Lemma}
\newtheorem{theorem}{Theorem}

\newtheorem{algorithm}{Algorithm}
\allowdisplaybreaks[1]

\DeclareMathOperator*{\argmin}{arg\,min}
\DeclareMathOperator*{\plim}{plim}

\author{Ryo Okui\thanks{NYU Shanghai, 1555 Century Avenue, Pudong, Shanghai, China, 200122;
		and Department of Economics, University of Gothenburg,
		P.O. Box 640, SE-405 30 Gothenburg, Sweden.
		Email: \href{mailto:okui@nyu.edu}{okui@nyu.edu}}\quad
	and Wendun Wang\thanks{Econometric Institute, Erasmus University Rotterdam, Burg. Oudlaan 50,
		3062PA, Rotterdam; and Tinbergen Institute. Email: \href{mailto:wang@ese.eur.nl}{wang@ese.eur.nl}}}

\title{Heterogeneous structural breaks in panel data models\thanks{
		Okui acknowledges financial support from the Japan Society for the Promotion of Science (JSPS) under KAKENHI Grant Nos.\!\! 16K03598 and 15H03329. Wang acknowledges the financial support of an Erasmus University Rotterdam fellowship. This project commenced when Okui was at Vrije Universiteit Amsterdam and Kyoto University. 
		The authors would also like to thank the guest editor (Tom Wansbeek), two anonymous referees, Otilia Boldea, Pavel Cizek, Qu Feng, Cheng Hsiao, Arthur Lewbel, Robin Lumsdaine, Elena Manresa, Liangjun Su, Martin Weidner and participants at various seminars and conferences 
 for their valuable comments and suggestions.}}
\date{November 2018}

\begin{document}
	\maketitle
	
	\begin{abstract}
		
		This paper develops a new model and estimation procedure for panel data that allows us to identify heterogeneous structural breaks.
		We model individual heterogeneity using a grouped pattern. 
		For each group, we allow common structural breaks in the coefficients. However, the number, timing, and size of these breaks can differ across groups.
		We develop a hybrid estimation procedure of the grouped fixed effects approach and adaptive group fused Lasso. 
		We show that our method can consistently identify the latent group structure, detect structural breaks, and estimate the regression parameters. Monte Carlo results demonstrate the good performance of the proposed method in finite samples. An empirical application to the relationship between income and democracy illustrates the importance of considering heterogeneous structural breaks.\\
		
		
		
	\end{abstract}

	\newpage
	\section{Introduction}
	Panel data sets have become increasingly popular in economics and finance, because they allow us to make use of variations over both time and individual dimensions in a flexible way.
	However, when we analyze panel data, it is important to take into account structural changes, such as
	financial crises, technological progress, economic transitions, etc.
	occurring during the time periods covered in the data.
	This is because these events may influence the relationships between economic variables, causing breaks in the parameters of panel data models.
	An important issue is that the time points and/or the impact of the changes are likely to vary significantly across individuals.	For example, despite the wide-ranging effects of the 1997 Asian financial crisis, which resulted in the restructuring of most Asian economies, both China and India were largely unaffected \citep{park&yang&shi&jiang2010, radelet&sachs2000}.
	
Similarly, for the recent European debt crisis after 2009, not all European countries have been affected.
	Even within the Eurozone, some countries fell into crisis much earlier than the others, and the impact of the crisis on the economic structure also differed across countries, with Central European countries appearing to be much less affected than the main victims in Southern Europe, such as Italy and Spain \citep{claeves&vasicek&2014}.
	Together, these observations suggest that structural breaks
	are not necessarily common across all observational units and breaks might occur at different time points across units and/or be of heterogeneous size.
	Importantly, existing break detection techniques for panel data
	primarily focus on only \emph{common} structural breaks.
	
	This paper provides a new model and estimation procedure that allows us to detect heterogeneous structural breaks.
	We consider a linear panel data model in which the coefficients are heterogeneous and time varying.
	 We model individual heterogeneity via a grouped pattern, and allow the group membership structure (i.e., which individuals belong to which group) to be unknown and estimated from the data.
	For each group, we then allow common structural breaks in the coefficients, while the number of breaks, breakpoints, and/or break sizes can differ across groups.

To estimate this model, we employ a hybrid procedure of the grouped fixed effects (GFE) method proposed by \citet{bonhomme&manresa2015} and the adaptive group fused Lasso (AGFL) in \citet{qian&su2016}.\footnote{Lasso stands for ``least absolute shrinkage and selection operator'', as introduced by \citet{tibshirani1996}.}
We title our procedure Grouped AGFL (GAGFL).
The idea of our estimation approach is to use GFE to estimate the group memberships, and simultaneously for each group to employ AGFL to detect the breaks and estimate the coefficients.
The GFE approach estimates the group memberships by minimizing the sum of squared residuals. We choose the GFE approach for classification because it facilitates theory and guarantees that all units are categorized to one of the groups.
For its part, AGFL estimates the break dates by minimizing a penalized least squares objective function, where the penalty term is the norm of the difference between the values of the time-varying parameters in adjacent time periods. It is particularly useful in our context because it estimates break dates jointly with the coefficients, and it allows breaks to occur in every period.

Of course, break detection, which is interesting in its own right, also improves the estimation of coefficients and group memberships by automatically pooling the time periods between two breaks. This hybrid estimation approach allows us to: 1) consistently estimate the latent group membership structure, 2) automatically determine the number of breaks and consistently estimate the breakpoints for each group in one joint step,
and furthermore 3) consistently estimate the regression coefficients with group-specific structural breaks. Computationally, we develop an iterative method to compute the estimates by combining the \textit{K-means} \citep{bonhomme&manresa2015} algorithm with Lasso. This algorithm works particularly well when the number of groups is not large, which is typically the case in many applications \citep{bonhomme&manresa2015,vogt2017classification,su&shi&phillips2016,lu&su2017}.

	Imposing a grouped pattern is a convenient way to model individual heterogeneity. Nevertheless, while in principle we may be able to analyze each individual separately by allowing complete individual heterogeneity,
	such an approach is not practical for two reasons. First, individual estimation can be inefficient as it does not make use of cross-sectional information at all, especially given that the length of time series in a typical panel data set is not very large \citep{wang&zhang&paap}. Second, analyzing individuals separately fails to capture any common pattern, which prevents learning from the experience of other individuals.
	A grouped pattern of heterogeneity is then especially useful because it allows us to estimate the coefficient parameters (and the breaks) more efficiently and thereby provides insights into the similarity of individual units. Moreover, it does not require a priori knowledge about the determinants of the group structure, and we can obtain an intuition about the underlying mechanism that causes heterogeneity by
	investigating the estimated group structure. Imposing a grouped heterogeneity pattern also enables us to estimate coefficients for every period, which further allows us to detect consecutive structural breaks. 

	We may conveniently characterize individual time-varying heterogeneity using a group pattern, such that individual units within a given group share the same time-varying paths, in many applications. For example, \citet{bonhomme&manresa2015} showed in an analysis of the democracy--income relationship that the time-varying unobserved heterogeneity demonstrates a group pattern and countries assemble
	depending on whether and when they experienced democratic transitions. Elsewhere, \citet{ando&bai2016} found that the styles of US mutual funds and asset return performance in mainland China both featured a group pattern of heterogeneity. In addition, \citet{hahn&moon2010} provided sound foundations for group structure in game theoretic or macroeconomic models where a multiplicity of equilibria is expected and \citet{bonhomme&lamadon&manresa2017} argued that group structure can be a good discrete approximation, even if individual heterogeneity is continuous.
	
	It is important to consider heterogeneity and structural breaks jointly. Ignoring heterogeneous structural breaks can lead to the incorrect detection of breakpoints and inconsistent slope coefficient estimates.
		The potential harm arising may be demonstrated in the following two scenarios. First, we consider pooling individuals with different breakpoints.\footnote{This includes the special case of pooling individuals with and without structural breaks.}
		In this case, smaller breaks characterize the pooled data. Hence, it is more difficult to detect breaks in finite samples, and the failure of break detection further leads to incorrect coefficient estimates.
		This issue is also highlighted by \citet{baltagi&qu&kao2016}, who showed via simulation that their break detection method (assuming common breaks) suffers from a significant loss of accuracy when some series contain no break.
		Second, even if all individuals have breaks at the same time, ignoring heterogeneity in break size and pooling all individuals may average out the break, and thus act against the correct detection of breaks, again leading to inconsistent coefficient estimates. Our proposed GAGFL method can simultaneously address both these problems because it allows us to identify which individuals are (un)affected by the breaks and to detect the group-specific breaks in one step.

	We examine the finite sample performance of our method via Monte Carlo simulation and compare it with other break detection techniques assuming common breaks.
	GAGFL performs well in heterogeneous panels with structural breaks occurring at possibly different points in finite samples.
	First, we are able to estimate group membership precisely, and the clustering accuracy improves as the length of the time series increases.
	Second, we are able to estimate the correct number of breaks and the true breakpoints for each group.
	We also find that ignoring heterogeneity in breaks (even if we account for the heterogeneity in coefficients) leads to an inconsistent estimation of the number of breaks, along with less accurate breakpoint estimates. This also results in inaccurate coefficient estimates.
	Importantly, because of the accurate estimation of both groups and breaks, our coefficient estimates are much more precise than those produced by common break detection methods.
	
	We illustrate our method by revisiting the relationship between income and democracy analyzed by
	\citet{acemoglu&johnson&robinson&yared2008}. We mainly compare our results with the GFE estimates produced by \citet{bonhomme&manresa2015} and find that our method provides a compatible but different grouped pattern.
	On the one hand, we can distinguish between countries with stable, early and late transition democracies, and this result is consistent with that of \citet{bonhomme&manresa2015}.
	On the other hand, we identify countries that have a fluctuating income--democracy relation toward the middle and/or end of the period not captured by GFE. The different grouped pattern arises because the groups formed by GAGFL draw on the whole coefficient vector that allows for structural changes. Thus, the grouped pattern produced by GAGFL not only reflects the magnitudes of the differences in the coefficient estimates, but also the heterogeneity in structural change.

	The remainder of this paper is organized as follows. Section \ref{sec:literature} reviews the related literature. Section \ref{sec:model} describes the model. Section \ref{sec:estimation} explains the estimation method and establishes its asymptotic properties. Section \ref{sec-fe} discusses the models with individual-specific fixed effects.
	In Section \ref{sec-extension}, we examine various extensions of our basic model and our estimation method and Section \ref{sec:simulation} considers the finite sample properties of our estimator via simulation.
	To demonstrate our method using actual data, we analyze the association between income and democracy in Section \ref{sec:empirical}.
	Section \ref{sec:conclusion} provides some concluding remarks. All technical details are in the Appendix, and a supplementary file contains additional theoretical results, further simulation studies, and an added application on the determinants of the savings rate.

	\section{Literature review}\label{sec:literature}
	
	Individual heterogeneity in panel data has been widely documented in empirical studies, especially heterogeneity in slope coefficients (see, for example, \citet{su&chen2013} and \citet{durlauf&kourtellos&minkin2001} for cross-country evidence and \citet{browning2007} for ample microeconomic evidence). Various estimation techniques have been developed to account for cross-sectional slope heterogeneity, including random coefficient models \citep{swamy1970}, mean group estimation \citep{pesaran&smith95}, GFE \citep{bonhomme&manresa2015}, classifier-Lasso \citep{su&shi&phillips2016}, and Panel-CARDS \citep{wang&phillips&su2016}, among others.
	
	Our work most closely relates to recent research concerning
latent group structure in panel data.\footnote{Note that \citet{bester&hansen2013} and \citet{ando&bai2015} also considered panel models with group structure, but they assume that the group structure is known.}
	\citet{sun2005} considered a finite-mixture model. \citet{lin&ng2012} and \citet{su&shi&phillips2016} examined cases in which the slope coefficients of individuals belong to an unknown latent group.
	\citet{bonhomme&manresa2015} introduced a
	GFE estimator that allows the fixed effect parameters to have a time-varying grouped pattern.
	Note that \citet{hahn&moon2010} discussed a similar estimator in the context of the multiple equilibria problem in empirical industrial organization, \citet{ando&bai2016} studied unobserved group factor structures, and \citet{vogt2017classification} considered nonparametric regression models with latent group structure.
However, these authors all assume that the panel is stable without structural change.
	We contribute to this literature by studying latent group estimation in the presence of unknown heterogeneous structural breaks.
	\citet{su&wang&jin2015} also considered grouped patterns with structural instability. They modeled instability using continuous time-varying slope coefficients, whereas we consider discontinuous structural breaks.
	
	While testing and dating structural breaks in panel data has become a popular research topic, many of the existing techniques employ the assumption of homogeneous slope coefficients.
	For example, \citet{bai2010} studied the estimation of a common break in means
	and variances for panel data, \citet{deWachter&tzavalis2012} developed a
	break detection test for panel autoregressive models, and \citet{kim2011} proposed an estimation procedure for a common
	deterministic time-trend break in large panels. 
		In other work, \citet{baltagi&kao&liu2017} introduced a procedure for breakpoint estimation
	in panel models for (non)stationary 
	regressors and error terms, and \citet{qian&su2016} proposed estimation procedures for panels with common breaks in slope coefficients using AGFL. \citet{li&qian&su} later extended this idea to interactive fixed effects models. We differ from these studies by considering heterogeneous structural breaks occurring for some individuals and/or with
different sizes.

	Despite ample evidence of cross-sectional heterogeneity in panel data, only a few studies have considered structural breaks in heterogeneous panels, which is the topic of this study. 
	A closely related work is \citet{baltagi&qu&kao2016}, which considered breakpoint estimation in heterogeneous panels with and without cross-sectional dependence.
	Our work differs in three main aspects.
	First, \citet{baltagi&qu&kao2016} considered structural breaks common to all individuals.
They documented in their simulation study that their break detection method loses significant accuracy if a break only affects a part of the individuals. We explicitly model heterogeneous breaks that only impact a proportion of \emph{unknown} units, and estimate the range of impacts from the data.
Incorporating heterogeneity in structural breaks also allows us to estimate the breakpoints more accurately than \citet{baltagi&qu&kao2016}. %
	Second, given the number of \emph{common} breaks, \citet{baltagi&qu&kao2016} estimated the breakpoints and then the slope coefficients in separate steps. In contrast, our method simultaneously estimates the number of breaks, breakpoints, and slope coefficients in one joint step.\footnote{To improve finite sample performance, we may re-estimate the slope coefficients after obtaining the breakpoints from the Lasso procedure.} Note that pooling the units may make it more difficult to determine the number of breaks in the presence of heterogeneity because it also dilutes the break sizes.
	Finally, we model individual heterogeneity via the grouped pattern, while \citet{baltagi&qu&kao2016} allowed for individual-specific coefficients. 
	We believe the parsimony of group specification in our approach may improve efficiency and achieve a faster convergence rate. We compare our method with \citet{baltagi&qu&kao2016} in the simulation to confirm this.

	In other work, \citet{pauwels&chan&griffoli2012} proposed a testing procedure for
	the presence of structural breaks that allows breaks to occur to only a proportion of individuals. 
	Our approach complements \citet{pauwels&chan&griffoli2012} by estimating the breakpoints and slope coefficients, and providing insights into which particular individuals are (un)affected by the breaks.
	Compared with individual time series estimation, we make use of cross-section variation and thus are able to estimate consecutive structural breaks. 
	For their part, \citet{schnatter&kaufmann} considered clustering multiple time series using the finite-mixture method in a Bayesian framework. They assume that the group memberships, characterized by the belonging probabilities, either depend on some exogenous observables through a specific (e.g., logistic) function or are equal to the group size. In contrast, we model group memberships via discrete indicators and allow them to be fully unrestricted.
	Finally, our work is also closely related to studies utilizing shrinkage estimators to estimate structural breaks (e.g. \citet{lee&seo&shin2016, qian&su2016, cheng&liao&schorfheide2016, li&qian&su}).
	However, none of these studies considered heterogeneous panels with possibly group-specific structural breaks.

	\section{Model setup}\label{sec:model}
	
	In this section, we describe the model setting.
	We consider a linear panel data model with heterogeneous and time-varying coefficients.
	The heterogeneity is restricted to have a grouped pattern, and structural breaks characterize the time-varying nature of the coefficients.
	The discussion of models with individual fixed effects is deferred to Section \ref{sec-extension}.
	
	Suppose that we have panel data $\{ \{ y_{it}, x_{it} \}_{t=1}^T\}_{i=1}^N$, where
	$y_{it}$ is a scalar dependent variable and $x_{it}$ is a $k\times 1$ vector of regressors, typically including the first element being 1.
	As usual, $t$ and $i$ denote time period and observational unit, respectively.
	The number of cross-sectional observations is $N$ and the length of the time series is $T$.
	We consider the following linear panel data model:
	\begin{equation*}
	y_{it}=x'_{it}\beta_{i,t}+\epsilon_{it},\quad i = 1,\ldots,N, \quad t = 1,\ldots, T,
	\end{equation*}
	where $\epsilon_{it}$ is the error term with zero mean.
	The unknown coefficient $\beta_{i,t}$ is heterogeneous across individuals and changes over time.
	We put a structure on $\beta_{i,t}$, namely a grouped pattern of heterogeneity and structural breaks,
	to make it estimable and this also facilitates interpretation of the model.
	
	We assume that $\beta_{i,t}$ is group specific.
	Suppose that observational units can be divided into $G$ groups.
	Let $\mathbb{G} = \{1, \dots ,G\}$ be the set of groups
	where $g_i \in \mathbb{G}$ indicates the group membership of unit $i$.
	Units in the same group share the same time-varying coefficient $\beta_{g,t}$, where $g \in \mathbb{G}$.
	Our model can be rewritten as
	\begin{equation}
	y_{it}=x'_{it}\beta_{g_i,t}+\epsilon_{it},\quad i = 1,\ldots,N, \quad t = 1,\ldots, T. \label{eq:model}
	\end{equation}
	We use this representation of our model hereafter. Imposing a grouped pattern offers a sensible and convenient way to model individual heterogeneity because it allows us to capture heterogeneity in a flexible way while keeping the model parsimonious, permitting us to take advantage of cross-sectional variation in break and coefficient estimation \citep{bai2010}.
	The group membership structure $\{g_i\}_{i=1}^N$ is unknown and to be estimated. 
	
	We also assume that for each group $g$, the time-varying pattern of coefficients $\{\beta_{g,1},\ldots,\beta_{g,T}\}$ can be characterized by structural breaks.
	For each group, there are $m_g$ breaks and $\mathcal{T}_{m_g, g} = \{T_{g,1}, \dots, T_{g,m_g} \}$ denotes a set of break dates, where both the number of breaks $m_g$ and break dates $\mathcal{T}_{m_g,g}$ are group specific and unknown, and we estimate them from the data.
	The value of coefficient $\beta_{g,t}$ changes only at a break date and remains the same in the period between any two break dates.
	Let $\alpha_{g,j}$, $j= 1, \dots, m_g$ be the value of coefficients until the $j$-th break date and $\alpha_{g,m_g+1}$ be the value of coefficients in the last period:
	\begin{align*}
	\beta_{g,t} =\alpha_{g,j}, \quad \text{ if } \quad T_{g,j-1} \le t < T_{g,j},
	\end{align*}
	where we define $T_{g,0} =1$ and $T_{g,m_g+1} = T+1$.
	Note that we allow consecutive structural breaks to occur in two adjacent periods because we make use of cross-sectional observations via grouping, and this clearly includes the end-of-sample breaks \citep[c.f.][]{andrews2003} as a special case.
	The true number of breaks for each group, $m_g^0$, is permitted to increase as $T$ increases, and the minimum break size ($\min_{g \in \mathbb{G}, 1 \le j \le m_g^0+1} \| \alpha_{g, j+1} - \alpha_{g, j}\|$) is allowed to contract to zero as $(N,T)\to\infty$.

	Our model of heterogeneous breaks permits identification of the number of groups and breaks under regularity conditions in the sense that 1) an ignored break cannot be incorporated by increasing the number of groups and; 2) distinct groups cannot be modeled as a homogeneous pool with structural breaks. 
	The identification of groups separately from breaks is relatively straightforward because we cannot incorporate structural breaks by increasing the number of groups. To illustrate the identification of structural breaks, we consider a simple case in which the model includes only an intercept $\beta_t$ with one break, and there are two groups which differ only in break dates, $T_{1,1} \neq T_{2,1}$: 
	$$
	y_{it}=\beta_{1} I(t<T_{j,1})+ \beta_{2} I(t\geq T_{j,1}) + \epsilon_{it},\quad \textrm{for $i$ such that $g_i = j$ and}\ j = 1,2.
	$$
	These two groups cannot be treated as one group with two breaks because observations between $T_{1,1}$ and $T_{2,1}$ correspond to distinct intercepts for the two groups. Next, we consider two groups with a break at a common date but of different sizes. This clearly leads to the separation of two groups we cannot merge because the coefficients are heterogeneous in at least one of the regimes. We generalize these arguments to cases with heterogeneous breaks at different points and of distinct sizes and to cases with multiple breaks. Therefore, we cannot accommodate the underestimation of the number of groups by increasing the number of breaks as long as groups are separated.

	This model is very general. If the parameter $\beta_{g,t}$ is constant over time, the model reduces to that considered in \citet{su&shi&phillips2016} and \citet{lin&ng2012}. If $\beta_{g_i,t}$ is homogeneous over individuals, it then reduces to panel models with common structural breaks as in \citet{qian&su2016} and \citet{baltagi&kao&liu2017}. It also includes the grouped fixed effect model by \citet{bonhomme&manresa2015} as a special case where only the intercept is allowed to have a time-varying grouped pattern.

	\section{Estimation method and its asymptotic properties}\label{sec:estimation}

	In this section, we explain our estimation method and consider its asymptotic properties.
	Our estimation method is a hybrid of GFE by \citet{bonhomme&manresa2015} and AGFL by \citet{qian&su2016}.

	We first introduce some notation.
	Let $\beta$ be the vector stacking all $\beta_{g,t}$ such that
	$\beta = (\beta_{1,1}' , \dots, \beta_{1,T}', \beta_{2,1}', \dots, \beta_{G,T}')$.
	Let $\mathcal{B} \subset \mathbb{R}^k$ be the parameter space for each $\beta_{g,t}$.
	The parameter space for $\beta$ is $\mathcal{B}^{GT}$.
	Let $\gamma$ be the vector of $g_i$s such that $\gamma = \{g_1, \dots, g_N\}$.
	Note that $\mathbb{G}^N$ is the parameter space for $\gamma$.

	\subsection{Estimation method}\label{sec:estimation-method}
	
	We estimate $(\beta, \gamma)$ by minimizing a penalized least squares objective function. 
	We employ an iterative procedure in which we iterate the estimation of $\gamma$ by minimizing the sum of squared errors for each unit and the estimation of $\beta$ by applying the AGFL to each group. 
	We state our estimation method assuming that $G$ is known in this section and discuss how to select the number of groups $G$ in Section \ref{sec-choose-tuning}.

	We propose to estimate model~\eqref{eq:model} by minimizing the following penalized least squares objective function:
	\begin{align}
	(\hat \beta, \hat \gamma)
	= \argmin_{(\beta, \gamma) \in \mathcal{B}^{GT} \times \mathbb{G}^N}
	\frac{1}{N T} \sum_{i=1}^N \sum_{t=1}^T
	(y_{it}  - x_{it} ' \beta_{g_i ,t} )^2
	+ \lambda \sum_{g \in \mathbb{G}} \sum_{t=2}^T
	\dot w_{g,t} \left\Vert \beta_{g, t} - \beta_{g, t-1}
	\right\Vert. \label{eq-objective-function}
	\end{align}
	The second term on the right-hand side of \eqref{eq-objective-function} is the penalty term, where $\lambda$ is a tuning parameter whose choice is discussed in Section \ref{sec-choose-tuning}, and $\dot w_{g,t}$ is a data-driven weight defined by
	\begin{align}
	\dot w_{g,t}
	= \left\Vert \dot \beta_{g, t} - \dot \beta_{g, t-1}\right\Vert^{-\kappa} \nonumber
	\end{align}
	with $\kappa$ being a user specific constant and
	$ \dot \beta $ being a preliminary estimate of $ \beta$.
	We can use the GFE-type estimate for $\dot \beta$ as follows:
	\begin{align}
	( \dot \beta , \dot \gamma)
	= \argmin_{(\beta, \gamma) \in \mathcal{B}^{GT} \times \mathbb{G}^N}
	\sum_{i=1}^N \sum_{t=1}^T
	(y_{it}  - x_{it} ' \beta_{g_i ,t} )^2.\label{eq-gfe}
	\end{align}
	The resulting coefficient and group membership estimates are both consistent (see Theorem~\ref{thm-gfe-ad} and Appendix~\ref{sec-gfe-proof}), although the coefficient estimates vary for each period because the least squares objective function in \eqref{eq-gfe} does not include the penalty term.\footnote{Note that, strictly speaking, these estimators are not covered by \citet{bonhomme&manresa2015} because the coefficients are heterogeneous and change over time.}

	Despite the fact that the estimated group memberships and slope coefficients are jointly obtained by solving the optimization problem~\eqref{eq-objective-function}, we can better understand this objective function by conditioning. Given the coefficients, the estimate of $g_i$ ($i$'s group membership) is a group that yields the smallest sum of squared residuals for $i$.
	Note that the second term in \eqref{eq-objective-function} does not depend on $g_i$ and does not have a direct effect on the estimation of $g_i$ (of course, it indirectly affects the estimation through the coefficients; see the more detailed discussion following Algorithm~\ref{al-kmeans}).
	
	Next, given the group membership structure, our estimation problem is the same as applying the AGFL by \citet{qian&su2016} to each group. The penalty term in \eqref{eq-objective-function} is proportional to the norm of the difference between the coefficients in adjacent periods.
	As is well known in the Lasso literature, this type of penalty term (called $L_1$ penalty) has a sparsity property and provides us estimates with the properties that $\hat \beta_{g,t} = \hat \beta_{g,t-1}$ for some $g$ and $t$.
	Estimated break dates are periods at which $\hat \beta_{g,t} - \hat \beta_{g,t-1} \neq 0$.
	Let $\hat{\mathcal{T}}_{\hat{m}_g,g}$ denote the set of estimated break dates for group $g$
	such that $\hat{\mathcal{T}}_{\hat{m}_g,g} = \{ t \in \{ 2, \dots, T \} \mid \hat \beta_{g,t} - \hat \beta_{g,t-1} \neq 0\}$, then we can estimate the number of breaks for group $g$ by the cardinality of $\hat{\mathcal{T}}_{\hat{m}_g,g}$.

The consistency of preliminary estimates $\dot{\beta}$ plays an important role in break detection as it yields appropriates weights ($\dot w_{g,t}$'s).
	Note that if $\beta_{g,t} - \beta_{g,t-1} =0$, then $\dot \beta_{g,t} - \dot \beta_{g,t-1} $ is likely to be close to zero and $\dot w_{g,t}$ is likely to be large, resulting in a heavy penalty. This in turn allows us to achieve the consistent estimation of break dates.
	The use of consistent estimators as initial values in our iterative algorithm given below facilitates the convergence. However, the preliminary coefficient estimate fails to capture the structural breaks, and as a result, the group and coefficient estimates produced by the GFE-type objective function~\eqref{eq-gfe} are also less accurate than those produced by \eqref{eq-objective-function}.

	To obtain the minimizer in \eqref{eq-objective-function}, we propose to use the following iterative algorithm.
	\begin{algorithm}
		\label{al-kmeans}
		Set $\gamma^{(0)}$ as the initial GFE estimate of grouping $\dot{\gamma}$, and $s = 0$.
		
		\begin{description}
			\item[Step 1]: For the given $\gamma^{(s)}$, compute
			\begin{align}
			\beta^{(s)} = \argmin_{\beta\in\mathcal{B}^{GT}}
			\frac{1}{N T} \sum_{i=1}^N \sum_{t=1}^T
			(y_{it}  - x_{it} ' \beta_{g^{(s)}_i ,t} )^2
			+ \lambda \sum_{g \in \mathbb{G}} \sum_{t=2}^T
			\dot w_{g,t} \left\Vert \beta_{g, t} - \beta_{g, t-1}
			\right\Vert. \label{eq-s-obj}
			\end{align}
			
			\item[Step 2]: Compute for all $i\in\{1,\ldots,N\}$,
			\begin{align*}
			g_i^{(s+1)}=\argmin_{g\in \mathbb{G} } \sum_{t=1}^T(y_{it}-x_{it}'\beta^{(s)}_{g,t})^2.
			\end{align*}
			\item[Step 3]: Set $s= s+1$. Go to Step 1 until numerical convergence.
		\end{description}
	\end{algorithm}
	
	Step 1 applies AGFL by \citet{qian&su2016} for each estimated group.
	Step 2 updates the group membership based on the least squares objective function.
	We assign each unit to one of the groups that gives the smallest sum of squared residuals.
	The penalty term does not contribute to the estimation of group membership directly because it does not depend on $i$.
	However, it does indirectly improve the group estimation of \eqref{eq-gfe} in finite samples, because it forces us to pool the time points between two breaks. Such pooling obviously makes use of more observations in estimating the slope coefficients, leading to more precise coefficient estimates and to a better group structure estimate. 
	
	We choose the preliminary GFE-type estimate of the grouped pattern as the initial grouping $\gamma^{(0)}$, namely $\gamma^{(0)}=\dot{\gamma}$.
	The estimation of $(\dot \beta, \dot \gamma)$ may be implemented by an algorithm similar to Algorithm \ref{al-kmeans}, in which the objective function in \eqref{eq-s-obj} is replaced by
	$\sum_{i=1}^N \sum_{t=1}^T 	(y_{it}  - x_{it} ' \beta_{g^{(s)}_i ,t} )^2 /(N T)$.
Given $\dot{\gamma}$ is consistent (see Appendix \ref{sec-gfe-proof}), this makes the convergence fast. 
	 Note that although individual-invariant regressors are in principle allowed in model~\eqref{eq:model} as long as there are periods without a structural break, this preliminary GFE-type estimation prevents including individual-invariant regressors (e.g., some global variables that are common to all countries) when $x_{it}$ includes a constant term, because these regressors are multicollinear with the constant. 
A possible solution is to replace the adaptive Lasso penalty by the SCAD penalty \citep{fan&li2001}, which does not require preliminary estimates, yet has similar properties to the adaptive Lasso.

	As a preliminary \emph{consistent} estimate is used as an initialization, this algorithm converges quickly, and the value of the objective function does not increase over the iterations.
	To implement GFE estimation of \eqref{eq-gfe}, we draw a large number of initial values at random and select the estimate that minimizes the objective function.
	The main computational effort involves trialing a large number of starting values for the preliminary GFE estimates, and the computation time increases linearly with this number.\footnote{In our simulation experiment with $(N,T)$=(100,40), one estimation based on 100 starting values takes roughly 1.4 seconds, and that based on 1,000 starting values takes roughly 12.5 seconds of CPU time.} This algorithm works well when the number of groups is not large and there are no outliers. When the number of groups is large, 
	a more robust difference-of-convex functions programming \citep{chu2017} can be considered.

	\subsection{Asymptotic theory}
	\label{subsec: asym}
	We derive the asymptotic properties of GAGFL.
	We first show that the difference between GAGFL and AGFL applied to each group under known group memberships is asymptotically small.
	As a result, the break dates can be estimated consistently by our method.
	The asymptotic distribution of the coefficient estimator is
	the same as the least squares estimator under known group memberships and known break dates.
	In this section, we use the following notation. Let $M$ denote a generic universal constant.
	Parameters with superscript 0 are the true values.
	
	We make the following assumptions to derive the asymptotic results.	The first set of assumptions (Assumptions \ref{a: basic}--\ref{a: tail}) is for GFE-type estimation.
	These assumptions are similar to those used in \citet{bonhomme&manresa2015}, but the difference arises because we consider more general models.

		\begin{assumption}\ \\
		\label{a: basic}
		\vspace*{-0.7cm}
		\begin{enumerate}
			\item 		\label{as-compact}
		$\mathcal{B}$ is compact.

			\item 		\label{as-exogenous}
		$E (\epsilon_{it} x_{it}) = 0 $ for all $i$ and $t$.
			\item		\label{as-ij-bound}
		There exists $M >0$ such that for any $N$ and $T$,
		\begin{align*}
		\frac{1}{N}\sum_{i=1}^N \sum_{j=1}^N
		\left| \frac{1}{T}
		 \sum_{t=1}^T
		E \left( \epsilon_{it} \epsilon_{jt}  x_{it}' x_{j t} \right)
		\right|
		< M .			
		\end{align*}

			\item 		\label{as-ijcov-bound}
		There exists $M >0$ such that for any $N$ and $T$,
		\begin{align*}
		 \left|
		\frac{1}{N^2} \sum_{i=1}^N \sum_{j=1}^N
		\frac{1}{T} \sum_{t=1}^T \sum_{s=1}^T
		Cov \left( \epsilon_{it} \epsilon_{j t}  x_{it}' x_{j t} ,
		\epsilon_{is} \epsilon_{j s}  x_{is}' x_{j s} \right)
		\right|
		< M .
		\end{align*}	
	
		\item
		\label{as-x-4moment}
		There exists $M >0$ such that for any $N$ and $T$,
		$ (NT)^{-1} \sum_{i=1}^N \sum_{t=1}^T E \left( \left\| x_{it} \right\|^4 \right) < M$.
		\end{enumerate}
	\end{assumption}

	Assumption \ref{a: basic}.\ref{as-compact} requires that the parameter space of slope coefficients is compact, as in most econometric literature.
	Assumption \ref{a: basic}.\ref{as-exogenous} states that the regressors are exogenous. Note that it does not exclude the cases with $E (\epsilon_{it} x_{is}) \neq 0$ for $t \neq s$, and thus predetermined regressors are allowed.
	Assumptions \ref{a: basic}.\ref{as-ij-bound} and \ref{a: basic}.\ref{as-ijcov-bound} restrict
	the magnitude of the variability and the degree of dependence in the data.
	For example, when the data are i.i.d.~over time, $\epsilon_{it}$ and $x_{it}$ are independent,
	and $\epsilon_{it} x_{it} $ has fourth-order moments,
	then these two assumptions are satisfied.
	Assumption \ref{a: basic}.\ref{as-x-4moment} imposes a condition on the fourth-order moment of $x_{it}$.

		\begin{assumption}\ \\
		\label{a: iden}
		\vspace*{-0.7cm}
		\begin{enumerate}
			\item
		\label{as-m-eigen}
		Let
		\begin{align*}
		M ( \gamma, g, \tilde g)
		= \frac{1}{N}
		\sum_{i=1}^N \mathbf{1} \left\{ g_i^0 = g \right\}
		\mathbf{1} \left\{ g_i = \tilde g \right\}
		\begin{pmatrix}
		x_{i1} x_{i1}' & 0 & \dots & 0 \\
		0 & x_{i2} x_{i2}' & \dots & \dots  \\
		\dots & \dots & \dots & 0 \\
		0 & \dots & 0 & x_{i T} x_{i T}'
		\end{pmatrix}.
		\end{align*}
		Let $\hat \rho ( \gamma , g , \tilde g )$
		be the minimum eigenvalue of $M(\gamma, g, \tilde g)$.
		There exist $\hat \rho $ and $\rho > 0$ such that $\hat \rho \to_p \rho$ and
		$\forall g$,
		$
		\min_{ \gamma \in \mathbb{G}^N}
		\max_{\tilde g \in \mathbb{G}}
		\hat \rho ( \gamma, g, \tilde g) > \hat \rho.
		$
			\item
		\label{as-d-bound}
		Let
		$
		D_{g\tilde g i}
		= 1/T \sum_{t=1}^T \left( x_{it}' ( \beta_{g, t}^0 - \beta_{\tilde g, t}^0 )\right)^2.
		$
		For all $g \neq \tilde g$, there exists a $c_{g, \tilde g} > 0$, such that
		$
		\plim_{N, T \to  \infty} 1/N \sum_{i=1}^N D_{g \tilde g i} > c_{g,\tilde g}
		$
		and for all $i$,
		$
		\plim_{ T  \to \infty} D_{g \tilde g i} > c_{g,\tilde g}.
		$
		\end{enumerate}
	\end{assumption}

	Assumption \ref{a: iden} is an identification condition.
	Assumption \ref{a: iden}.\ref{as-m-eigen} provides an identification condition for the coefficients. Roughly speaking, this assumption states that there is no multicollinearity problem in any group structure.
	Assumption \ref{a: iden}.\ref{as-d-bound} provides an identification condition for group membership and may be called the ``group separation condition.''
	$ D_{g\tilde g i}$ is a measure of the distance between predicted values of $y_{it}$ under group $g$ and $\tilde g$ for unit $i$ and the conditions state that it is bounded away from zero.

	\begin{assumption}\ \\
	\label{a: tail}
	\vspace*{-0.7cm}
	\begin{enumerate}
			\item
		\label{as-ex2-tail}
		There exists a constant $M_{ex}^*$ such that as $N, T \to \infty$, for all $\delta > 0$, \\ $ 		\sup_{1 \le i \le N} 
		\Pr \left(   \sum_{t=1}^T \left\Vert \epsilon_{it} x_{it} \right\Vert^2 /T \ge M_{ex}^*\right)
		= O ( T^{-\delta})$.
		
			\item
		\label{as-x4-tail}
		There exists a constant $M_{x}^*$ such that as $N, T \to \infty$, for all $\delta >0$, \\ $		\sup_{1 \le i \le N}
		\Pr \left(  \sum_{t=1}^T \left\Vert x_{it} \right\Vert^4 /
	 T \ge M_x^* \right)
		= O ( T^{-\delta})$.
		
			\item
		\label{as-xb-mixing}
		There exist constants $a > 0$ and $d_1 > 0$ and a sequence $\alpha [t] < \exp(-a t^{d_1}) $ such that, for all $i= 1, \dots, N$ and $(g,\tilde g) \in \mathbb{G}^2$
		such that $g\neq \tilde g$, $\{ x_{it}'( \beta_{\tilde g, t}^0 - \beta_{g,t}^0)\}_t$,
		$\{ x_{it}'( \beta_{\tilde g, t}^0 - \beta_{g,t}^0) \epsilon_{it}\}_t$ are strongly mixing processes with mixing coefficients $\alpha [t]$.
		Moreover, $E (x_{it}'( \beta_{\tilde g, t}^0 - \beta_{g,t}^0) \epsilon_{it} ) =0$.
		
			\item
		\label{as-xb-tail}
		There exist constants $b_x>0$, $b_e >0$, $d_{2x} >0$ and $d_{2e}$ such that $\Pr ( | x_{it}'( \beta_{\tilde g, t}^0 - \beta_{g,t}^0) | > m ) \le \exp (1- (m/b_x))^{d_{2x}}$ and
		$\Pr ( | x_{it}'( \beta_{\tilde g, t}^0 - \beta_{g,t}^0) \epsilon_{it} | > m ) \le \exp (1- (m/b_e))^{d_{2e}}$, for any $i$, $t$ and $m >0$.
	\end{enumerate}
	\end{assumption}
	Assumption \ref{a: tail} is used to bound the maximum clustering error by placing restrictions on the tail behavior of the variables and the dependence structure. For example, Assumption \ref{a: tail}.\ref{as-x-4moment} may be violated when the tail of the distribution of $x_{i t}$ is so heavy that it does not possess a finite fourth moment. Another form of violation occurs when $x_{i t}$ exhibits a unit root. 	Assumption \ref{a: tail}.\ref{as-xb-mixing} requires that the mixing coefficient $\alpha [t]$ decays exponentially fast. This assumption enables us to apply exponential inequalities in \citet[Lemma B.5]{bonhomme&manresa2015} (which is based on \citet{rio2017asymptotic}).

	To examine the impact of the estimation error in the group structure on the coefficient estimates, we define 
	\begin{align*}
	\mathring \beta
	= \arg\min_{\beta \in \mathcal{B}^{GT}}
	\left( \frac{1}{N T}
	\sum_{i=1}^N \sum_{t=1}^T
	(y_{it}  - x_{it} ' \beta_{g_i^0 ,t} )^2
	+ \lambda  \sum_{g\in \mathbb{G}}
	\sum_{t=2}^T \dot w_{g,t} \left\Vert \beta_{g, t} - \beta_{g, t-1}
	\right\Vert \right) .
	\end{align*}
	Note that $\mathring \beta$ is the estimator of $\beta$ when the group memberships (i.e., $\gamma^0$) are known.
	Denote $N_g$ as the number of units in group $g$, i.e. $N_g = \sum_{i=1}^N \mathbf{1} \{ g_i^0 = g\}$ for $g \in \mathbb{G}$. 
	With the assumptions stated above, Lemma \ref{lem-db-cb-h}
	states that the impact of the estimation error in the group structure is limited, and thus the difference between $\hat \beta$ and $\mathring \beta$ is small.
	\begin{lemma}
		\label{lem-db-cb-h}
		Suppose that Assumptions \ref{a: basic}, \ref{a: iden}, and
		\ref{a: tail} are satisfied. Suppose also that $N_g / N \to \pi_g$ for some $0 < \pi_g < 1 $ for all $g\in \mathbb{G} $.  As $N,T \to \infty$, for any $\delta >0$,
		it holds that
		\begin{align}\label{eq:lemma1}
		\hat \beta_{g,t} = \mathring \beta_{g,t} + o_p (T^{-\delta}),
		\end{align}
		for all $g$ and $t$.
	\end{lemma}
	Note that $\delta$ in the theorem can be arbitrarily large, and we obtain this ``super-consistency'' result because we model heterogeneity via discrete grouped pattern with the number of groups being fixed and finite. This also enables us to detect the breaks even when the true group memberships are unknown. 
	
	Next, we show that our method detects breaks correctly.
	Recall that $m_g$ denotes the number of breaks for group $g$ 
	and $m_g^0$ is the true number of breaks for this group.
	$\mathcal{T}_{m_g,g} = \{ T_{g,1}, \dots, T_{g ,m} \}$ denotes a set of break dates
	and $\mathcal{T}_{m_g^0 ,g}^0 = \{T_{g,1}^0, \dots, T_{g, m_g^0}^0\}$ is the set of true break dates.
	$\mathcal{T}_{m_g^0, g}^{0c}$ denotes the complement of $\mathcal{T}_{m_g^0 ,g}^0$, representing the set of time with no breaks.
	Note that $\beta_{g,t}^0 - \beta_{g, t-1}^0 =0$ if $t \in \mathcal{T}_{m_g^0, g}^{0c}$.
	Let $\alpha_{g,1} = \beta_{g,1}$ and $\alpha_{g,j} = \beta_{g, T_{g,j-1}}$ for $j = 2, \dots, m_g^0+1$.
	Let
	$
	J_{\min} = \min_{g \in \mathbb{G}, 1 \le j \le m_g^0} \| \alpha_{g, j+1} - \alpha_{g, j}\|
	$
	be the minimum break size.
	We make the following assumption that is similar to Assumption A.2 in \citet{qian&su2016}.
	\begin{assumption}\ \\
	\label{a: break}\vspace*{-0.7cm}
	\begin{enumerate}
			\item
		\label{as-jmin}
		$ \sqrt{ N} T\lambda \left(\sum_{g\in \mathbb{G}} m_g^0\right)J_{\min}^{-\kappa}  = O_p (1)$.

			\item
		\label{as-lambda}
		$\sqrt{N} T \lambda N^{-\kappa /2} \to_p \infty$.
		
			\item
		\label{as-njmin}
		$\sqrt{N} J_{\min} \to \infty$.
	\end{enumerate}
	\end{assumption}
	Assumption \ref{a: break}.\ref{as-jmin} allows the total number of breaks of all groups $\sum_{g\in\mathbb{G}}m_g^0$ to diverge to infinity
	at a slow rate. Assumptions \ref{a: break}.\ref{as-jmin} and \ref{a: break}.\ref{as-lambda} are used to show the consistent detection of breaks by AGFL. 
	Assumption \ref{a: break}.\ref{as-njmin} states that break sizes are not very small so that all breaks can be identified.
	Note that it still allows the break size to tend to zero as long as the rate is slower than $\sqrt{N}$.
	
	In the following, we show that our method can consistently select the number of breaks and estimate the breakpoints even when the group membership is unknown. Consistent break detection in AGFL requires that the weights be adaptive and correctly estimated, which further requires the preliminary estimates used to construct the weights to be consistent. Hence, we first show that the preliminary estimator is $\sqrt{N}$-consistent.
		\begin{theorem}
		\label{thm-gfe-ad}
		Suppose that Assumptions \ref{a: basic}, \ref{a: iden}
 and \ref{a: tail} hold. Suppose that $N_g /N \to \pi_g > 0$ for any $g \in \mathbb{G}$.
		Then, it follows that for all $g$ and $t$,
		\begin{align}
			\dot \beta_{g,t} - \beta_{g,t}^0 = O_p\left(\frac{1}{\sqrt{N}}\right).
		\end{align}
	\end{theorem}
This theorem guarantees that the weight in the penalty, $\dot w_{g,t}$, becomes large when there is no break (i.e., $\beta_{g,t}-\beta_{g,t-1}$), which further enables us to detect periods without breaks. Let $\hat \theta_{g,t} = \hat \beta_{g,t} - \hat \beta_{g,t-1}$. 
	\begin{theorem}
		\label{lem-break-c-agfl}
		Suppose that Assumptions \ref{a: basic}, \ref{a: iden}, \ref{a: tail}, and \ref{a: break} hold. Suppose that $N_g /N \to \pi_g > 0$ for any $g \in \mathbb{G}$.
		It follows that
		\begin{align*}
		\Pr \left( \left\| \hat \theta_{g,t} \right\| = 0, \forall t \in \mathcal{T}_{m_g^0,g}^{0c}, g \in \mathbb{G}\right) \to 1
		\end{align*}
		as $N ,T \to \infty$ with $N /T^\delta \to 0$ for some $\delta$.
	\end{theorem}
	This theorem states that our method consistently spots dates on which there is no break. 
	This theorem requires that $N / T^{\delta} \to 0$ for some $\delta$. As $\delta$ is arbitrary, this condition is satisfied as long as $N$ is of geometric order $T$ but it does not hold if $N$ is of exponential order $T$.
	Note that the probability in this theorem is unconditional in the sense that it does not depend on knowing the true group membership. We also establish selection consistency as below.
		
	\begin{theorem}
		\label{lem-agafl-date}
		Suppose that Assumptions \ref{a: basic}, \ref{a: iden},
		\ref{a: tail}, and
		\ref{a: break} hold. Suppose that $N_g /N \to \pi_g > 0$ for any $g \in \mathbb{G}$.
		It holds that, as $N, T \to \infty$ with $N/T^\delta \to 0$ for some $\delta >0$,
		\begin{align*}
		\Pr (\hat m_g = m_{g}^0, \forall g \in \mathbb{G}) \to 1,
		\end{align*}
		and
		$
		\Pr \left( \hat T_{g,j} = T_{g,j}^0, \forall j \in \{1, \dots, m_g^0\}, g \in \mathbb{G} \mid \hat m_g =m_g^0, \forall g \in \mathbb{G}\right) \to 1.
		$
	\end{theorem}
	
	This theorem states that the number of breaks and break dates are estimated consistently as $N$ and $T$ go to infinity at an appropriate rate. As above, the probabilities are both unconditional, implying that the selection consistency is established even when we do not know the group structure. 
	
	Fundamentally, Theorems~\ref{lem-break-c-agfl} and~\ref{lem-agafl-date} both result from the super-consistency of group membership estimation and consistency of AGFL applied to each group with known group memberships. The error introduced by estimating group memberships is negligible (Lemma~\ref{lem-db-cb-h}) and GAGFL is asymptotically equivalent to applying AGFL to each group with known group structure. The AGFL estimator applied to each \emph{true} group is pointwise consistent with the convergence rate of $1/\sqrt{N}$ (see Lemma~\ref{lem-lasso-c} in Appendix~\ref{sec-agfl-proof} or Theorem 3.2(ii) of \citet{qian&su2016}). We can thus obtain the pointwise consistency of our GAGFL estimator under unknown group structure.

	Lastly, we present the asymptotic distribution of $\hat \beta$. Put simply, this is the same as that of the least squares estimator for each group within each no-break regime with known group memberships and break dates. 
		To state the theorem, we introduce the following notation.
	Denote 
	$I_{g,j}$ as the number of time periods between $T_{g,j}^0$ and $T_{g,j+1}^0$ for group $g$, and denote $I_{\min}$ as the minimum number of periods for which there is no break, namely $ 	I_{\min} = \min_{g \in \mathbb{G}, 1 \le j \le m_g^0+1} \| T_{g,j}^0 - T_{g,j-1}^0 \| $.
	Let
	$
	\Sigma_{x,g,j} = \plim_{T, N_g \to \infty} 1/(N_gI_{g, j})\sum_{g_i^0 = g}
	\sum_{t=T_{g,j}^0}^{T_{g,j+1}^0 -1} x_{it} x_{it}',
	$
	and
	\begin{align*}
	\Omega_{g, h, j, j'} = \lim_{T, N_g , N_h \to \infty}
	\frac{1}{\sqrt{N_g} \sqrt{N_h}} \frac{1}{\sqrt{I_{g, j}} \sqrt{I_{h, j'}} }
	\sum_{g_i^0 = g} \sum_{g_{i'}^0 = h}
	\sum_{t=T_{g,j}^0}^{T_{g,j+1}^0 -1}
	\sum_{t'=T_{h,j'}^0}^{T_{h,j'+1}^0 -1}E( x_{i t} x_{i' t'}' \epsilon_{i t} \epsilon_{i' t'}).
	\end{align*}
	Let $\Omega_{g,g}$ be a $(m_g^0+1) k \times (m_g^0+1) k $ matrix whose $(j,j')$-th $k\times k$ block is $\Omega_{g, g, j', j'}$,
	and $\Sigma_{x,g}$ be a $ (m_g^0 +1) k \times (m_g^0+1) k$ block diagonal matrix whose $t$-th diagonal block is $\Sigma_{x,g ,j}$.
	Further, let $\Omega$ be a $\sum_{g=1}^G (m_g^0 +1) k  \times \sum_{g=1}^G (m_g^0 +1) k  $ matrix whose $(g,h)$-th $ (m_g^0+1)k\times (m_h^0+1)k$ block is $\Omega_{g, h}$,
	and $\Sigma_{x}$ be a $ \sum_{g=1}^G (m_g^0 +1) k  \times \sum_{g=1}^G (m_g^0 +1) k $
	block diagonal matrix whose $g$-th diagonal block is $\Sigma_{x,g}$.
	We require the following extra assumptions for the asymptotic distribution.
	\begin{assumption}
		\label{as-cb-break}
		Suppose that $\Sigma_{x} $ and $\Omega$ are well defined, their minimum eigenvalues are bounded away from zero, and their maximum eigenvalues are bounded uniformly over $T$.
		$N_g /N \to \pi_g > 0$ for any $g \in \mathbb{G}$.
		Let
		\begin{align*}
		d_{g, NT} =  \frac{1}{\sqrt{N_g}}
		\sum_{g_i^0 =g} \left( \sum_{t=1}^{T_{g,1}^2-1 } x_{it} \epsilon_{it}/\sqrt{I_{g,1}} ,\dots, \sum_{t=T_{g, m_g^0}^0}^{T}x_{i t} \epsilon_{i t} / \sqrt{I_{g, m_g^0+1}}\right)'.
		\end{align*}
		For an $l \times \sum_{g=1}^G (m_g^0 +1) k$ matrix $D$, where $l$ does not depend on $T$ and $\lim_{T \to \infty} D \Omega D'$ exists and is positive definite,
		$D ( d_{1, N T}' , \dots, d_{G, N T}')' \to_d N (0, \lim_{T \to \infty} D\Omega D')$.
	\end{assumption}
	\begin{assumption}
		\label{as-jmin-imin}
		$N \sum_{g=1}^G (m_g^0) \lambda^2 I_{\min}^{-1} J_{\min}^{-2\kappa} =o_p (1)$.
	\end{assumption}

	Assumption \ref{as-cb-break} simply states that the standard assumptions for least squares are satisfied for each group and each span of periods between two breaks. The purpose of introducing $D$ is to analyze a finite dimensional vector of linear combinations of elements of $\hat \alpha$. Note that the dimension of $\hat \alpha$ is potentially increasing in $T$ and its asymptotic distribution is hard to discuss. We instead examine finite dimensional objects. Assumption \ref{as-jmin-imin} is a technical assumption on the break sizes.

	The following shows that the GAGFL slope coefficient estimator of group $g$ in regime $j$ is consistent and asymptotically normal with the convergence rate of $\sqrt{N_gI_{g,j}}$.
	
	\begin{theorem}
		\label{thm-mrb-ad-agfl}
		Suppose that Assumptions \ref{a: basic}, \ref{a: iden},
		\ref{a: tail},
		\ref{a: break}, \ref{as-cb-break}, and \ref{as-jmin-imin} hold. Suppose that $N_g /N \to \pi_g > 0$ for any $g \in \mathbb{G}$.
		Let $A$ be a diagonal matrix whose diagonal elements are \\
		$(I_{1,1}, \dots, I_{1,m_1^0+1}, I_{2,1}, \dots, I_{2,m_2^0 +1}, I_{3, 1} \dots, I_{G-1, m_{G-1}^0+1}, I_{G, 1}, \dots, I_{G,m_G^0+1})$.
		Let $\Pi$ be a $\sum_{g=1}^G (m_g^0 +1)  k \times \sum_{g=1}^G (m_g^0 +1) k $
		block diagonal matrix whose $g$-th diagonal block is an $(m_g^0+1)k \times (m_g^0 +1)k$ diagonal matrix with the elements being $\pi_g$.
		
		Conditional on $\hat m_g = m_g^0$ for all $g \in \mathbb{G}$,
		we have, if $(\max_{g\in \mathbb{G}} m_g^0 )^2 / (I_{\min} \min_{g\in \mathbb{G}} N_g) \to 0 $,
		\begin{align*}
		D \sqrt{N} A^{1/2}(\hat \alpha - \alpha^0)
		\to_d N( 0, D\Sigma_{x}^{-1} \Pi^{-1/2} \Omega \Pi^{-1/2} \Sigma_x^{-1} D').
		\end{align*}
	\end{theorem}

 We point out that these theoretical results also hold if we replace $\hat \beta$ with $\beta^{(0)}$ from Algorithm~\ref{al-kmeans}. This is because the initial group assignment in the algorithm, $\dot \gamma$, is consistent. Note that for $\hat \beta$, the consistency of $\dot \gamma$ is not fundamental (although it does provide rapid numerical convergence). What is crucial is the consistency of $\dot \beta $, because it guarantees that $\dot w_{g,t}$'s have appropriate orders, which makes it possible to detect breaks consistently. In this paper, we focus on $\hat \beta$ because we find that the iterative estimator $\hat \beta_{g,t}$ typically outperforms $\beta^{(0)}$ in finite samples due to possible refinement.

	\subsection{Choosing the number of groups and the tuning parameter for the Lasso penalty}
	\label{sec-choose-tuning}
	
	To implement GAGFL, we need to choose the tuning parameter in AGFL for each group, and at the same time specify the number of groups. We discuss these two issues in turn.

	First, to select the tuning parameter $\lambda$ of the Lasso penalty in AFGL, we follow \citet{qian&su2016} to minimize the following information criterion (IC):
	$$
	IC(\lambda)=\frac{1}{NT}\sum_{j=1}^{m+1}\sum_{t=T_j-1+1}^{T_j}\sum_{i=1}^N(y_{it}-x'_{it}\hat{\alpha}_{g_i,j})^2+\rho_{NT}k(m_{\lambda}+1),
	$$
	where $\hat{\alpha}_{g_i,j}$ is the post-Lasso estimate of the coefficients for each of $g$ and $j$, $m_\lambda$ is the number of breaks associated with the tuning parameter $\lambda$, $\rho_{NT}$ determines the amount of penalty on the number of breaks, and we choose $\rho_{NT}=c\ln(NT)/\sqrt{NT}$ with $c=0.05$ following \citet{qian&su2016}. We verify via simulation and applications that the performance of our method is robust to the choice of $c$ as long as it lies in a reasonable range. 
	
	Next, we discuss how to choose the number of groups. 
	In this paper, we focus on using an information criterion to determine the number of groups, $G$.\footnote{The literature also suggests a Lagrange multiplier test \citep{lu&su2017} to determine $G$.} We follow \citet{bonhomme&manresa2015} to consider the following Bayesian information criterion (BIC):
	\begin{equation}\label{eq:BIC}
	BIC(G)=\frac{1}{NT}\sum_{j=1}^{m+1}\sum_{t=T_j-1+1}^{T_j}\sum_{i=1}^N(y_{it}-x'_{it}\hat{\alpha}_{g_i,j})^2+\hat{\sigma}^2\frac{n_p(G)+N}{NT}\ln NT,
	\end{equation}
	where $\hat{\sigma}^2$ is a scaling parameter and can be obtained by an estimate of the variance of $\epsilon_{it}$, and $n_p(G)$ is the total number of estimated coefficients.
	The information criterion represents a tradeoff between model fitness and the number of parameters. One caution is that this tradeoff is more complicated in our model because increasing $G$ does not always lead to a larger number of parameters and a better fit in our case. It is actually possible that a larger value of $G$ results in fewer breaks in each group, so that the total number of parameters decreases and the model fit worsens. 
	Nevertheless, we find that more parameters correspond to a better fit in our simulation experiments. Furthermore, we can show that with an appropriate choice of the tuning parameter in the penalty terms, the model that coincides with the data generation process produces the lowest information criterion value.
	
	In principle, it is possible to replace $\hat{\alpha}_{g_i,j}$ in \eqref{eq:BIC} by initial estimates of coefficients ($\dot \beta$) that are fully time varying as defined in \eqref{eq-gfe}. 
	An advantage of using initial estimates to compute the BIC is that both the number of parameters and model fit are monotonically increasing in $G$. 
	However, the disadvantage is that less efficient coefficient estimates may result in less accurate selection of $G$ in finite samples. Simulation results (provided in Section S.1.1 of the supplement) indicate that the BIC based on the final estimates outperforms that based on the initial estimates.

	\section{Fixed effects model}
	\label{sec-fe}

	In this section, we consider a model with individual-specific fixed effects. The estimation uses first-differenced data and requires separate theoretical analysis. Here, we provide the assumptions and summarize the theoretical properties, while the proofs and additional details are in the supplement.
		
	In addition to the time-varying grouped fixed effect, we allow an additive time-invariant individual fixed effect $\alpha_i$ arbitrarily correlated with the regressors. We consider the following model:
	\begin{equation*}
	y_{it}=\alpha_i+x'_{it}\beta_{g_i,t}+\epsilon_{it},\quad i = 1,\ldots,N, \quad t = 1,\ldots, T.
	\end{equation*}
	The individual fixed effects, $\alpha_i$, can be eliminated by first differencing:
	\begin{equation}\label{eq:first-differece}
	\Delta y_{it}=x_{it}'\beta_{g_i,t}-x_{i,t-1}'\beta_{g_i,t-1}+\Delta\epsilon_{it},
	\end{equation}
	where $\Delta y_{it}=y_{it}-y_{it-1}$ for $i=1,\ldots,N$ and $t=2,\ldots, T$.
	The model is estimated by applying the GAGFL method to the transformed data.
	{Note that in our setting, it is possible to identify fully time-varying coefficients $\beta_{g,t}$, even though there are only $T-1$ time periods available because of differencing. To see this, suppose that group membership is known. Then, \eqref{eq:first-differece} can be understood as a cross-sectional regression model with dependent variable $\Delta y_{it}$ and regressors $(x_{it}', x_{i,t-1}')'$ for a given $t$. Both $\beta_{g,t}$ and $\beta_{g,t-1}$ can be estimated from the cross-sectional regression. The actual identification argument can be more involved because group memberships are unknown. Nevertheless, as shown by this rough argument, the availability of cross-sectional information is the key to identifying the coefficients even when they are fully time varying.}
	
	The GAGFL estimator for fixed effects models is defined as
	\begin{align*}
	(\hat \beta, \hat \gamma)
	= \argmin_{(\beta, \gamma) \in \mathcal{B}^{GT} \times \mathbb{G}^N}
	 \frac{1}{N T}
	\sum_{i=1}^N \sum_{t=2}^T
	(\Delta y_{it} - x_{it} '\beta_{g_i ,t} + x_{i,t-1}' \beta_{g_i, t-1} )^2
	+ \lambda  \sum_{g\in \mathbb{G}}
	\sum_{t=2}^T
	\dot w_{g,t} \left\Vert \beta_{g, t} - \beta_{g, t-1}
	\right\Vert.
	\end{align*}
	The results of Section \ref{subsec: asym} are not directly applicable to fixed effects models. For example, both $\beta_{g,t}$ and $\beta_{g,t-1}$ enter the equation and we cannot analyze two time periods separately even when there is a break between these two periods. Nevertheless, we can extend our analysis to fixed effects models. We first state the modified assumptions needed. 

\begin{assumption}\ \\
	\label{a: fe-basic}\vspace*{-0.5cm}
	\begin{enumerate}
		\item 
		\label{as-fe-compact}
		$\mathcal{B}$ is compact.
		
		\item 
		\label{as-fe-exogenous}
		$E(\Delta \epsilon_{it} x_{it} ) = E(\Delta \epsilon_{it} x_{i,t-1} ) =0 $
		for all $i$ and $t$.
		
		\item
		\label{as-fe-ij-bound}
		$1/N \sum_{i=1}^N \sum_{j=1}^N
		\left|
		1/T\sum_{t=2}^T
		E \left(\Delta  \epsilon_{it} \Delta \epsilon_{j t}  x_{it}' x_{j t} \right)
		\right|
		< M.
		$
		
		\item
		\label{as-fe-ijcov-bound}
		$
		\left|
		1/N^2 \sum_{i=1}^N \sum_{j=1}^N
		1/T\sum_{t=2}^T \sum_{s=2}^T
		Cov \left( \Delta \epsilon_{it} \Delta \epsilon_{j t}  x_{it}' x_{j t} ,
		\Delta \epsilon_{is} \Delta \epsilon_{j s}  x_{is}' x_{j s} \right)
		\right|
		< M.
		$
		
		\item	
		\label{as-fe-x-4moment}
		There exists $M >0$ such that for any $N$ and $T$,
		$
		1/(NT) \sum_{i=1}^N \sum_{t=1}^T E (\left\| x_{it} \right\|^4 )< M.
		$		
	\end{enumerate}
\end{assumption}
This assumption can be implied by Assumption~\ref{a: basic} except for \ref{a: fe-basic}.\ref{as-fe-exogenous}. 
However, we continue to make this assumption to facilitate our theory.

\begin{assumption}\ \\
	\label{a: fe-iden}\vspace*{-0.5cm}
	\begin{enumerate}
		\item 
		\label{as-fe-m-eigen}
		Let
		\begin{align*}
		M^F ( \gamma, g, \tilde g)
		=& \frac{1}{N}
		\sum_{i=1}^N \mathbf{1} \left\{ g_i^0 = g \right\}
		\mathbf{1} \left\{ g_i = \tilde g \right\} \\
		& \times
		\begin{pmatrix}
		x_{i1} x_{i1}' & -x_{i1} x_{i2}' & 0 & \dots & 0 \\
		-x_{i2} x_{i1} & 2 x_{i2} x_{i2}' & - x_{i2} x_{i3} ' & \dots & \dots \\
		0 & - x_{i3} x_{i2}  & \dots & \dots & 0 \\
		\dots & \dots & \dots & 2 x_{i,T-1} x_{i,T-1}'&  - x_{i,T-1} x_{i T} \\
		0 & \dots  & 0 & - x_{i T} x_{i,T-1} & x_{i T} x_{i T}'
		\end{pmatrix}.
		\end{align*}
		Let $\hat \rho^F ( \gamma , g , \tilde g )$
		be the minimum eigenvalue of $M^F(\gamma, g, \tilde g)$.
		There exists a $\hat \rho^F $  such that $\hat \rho^F \to_p \rho^F >0$ and
		$\forall g$,
		$
		\min_{ \gamma \in \mathbb{G}^N}
		\max_{\tilde g \in \mathbb{G}}
		\hat \rho^F ( \gamma, g, \tilde g) > \hat \rho^F.
		$
		
		\item
		\label{as-fe-d-bound}
		Let
		$
		D_{g\tilde g i}^F
		= 1/T \sum_{i=1}^N \sum_{t=2}^T \left( x_{it}'( \beta_{g, t}^0 - \beta_{\tilde g, t}^0) - x_{i,t-1}' ( \beta_{g, t-1}^0 - \beta_{ \tilde g, t-1}^0) 
		\right)^2. 
		$
		For all $g \neq \tilde g$, there exists a $c_{g, \tilde g}^F > 0$ such that
		$
		\plim_{N, T \to  \infty} 1/N \sum_{i=1}^N D_{g \tilde g i}^F > c_{g,\tilde g}^F
		$
		and for all $i$,
		$
		\plim_{ T  \to \infty} D_{g \tilde g i}^F > c_{g,\tilde g}^F.
		$
		
	\end{enumerate}
\end{assumption}

\begin{assumption}\ \\
	\label{a: fe-tail}\vspace*{-0.5cm}
	\begin{enumerate}
		\item 
		\label{as-fe-ex2-tail}
		There exists a constant $M_{ex}^*$ such that as $N, T \to \infty$, for all $\delta > 0$, 	\\	$\sup_{1 \le i \le N}
		\Pr \left(   \sum_{t=2}^T \left\Vert \Delta \epsilon_{it} x_{it} \right\Vert^2 /T \ge M_{ex}^* \right)
		= O ( T^{-\delta}) $,
		and \\
		$
		\sup_{1 \le i \le N}
		\Pr \left(\sum_{t=2}^T \left\Vert \Delta \epsilon_{it} x_{i,t-1} \right\Vert^2  /T \ge M_{ex}^* \right)
		= O ( T^{-\delta}).
		$
		
		\item
		\label{as-fe-x4-tail}
		There exists a constant $M_{x}^*$ such that as $N, T \to \infty$, for all $\delta >0$, \\ $		\sup_{1 \le i \le N}
		\Pr \left( \sum_{t=1}^T \left\Vert x_{it} \right\Vert^4  /T \ge M_x^* \right)
		= O ( T^{-\delta})$.
		
		\item
		\label{as-fe-xb-mixing}
		There exist constants $a > 0$ and $d_1 > 0$ and a sequence $\alpha [t] < \exp(-a t^{d_1}) $ such that, for all $i= 1, \dots, N$ and $(g,\tilde g) \in \mathbb{G}^2$
		 $g\neq \tilde g$, $\{ x_{it}'( \beta_{\tilde g, t}^0 - \beta_{g,t}^0) -x_{i,t-1}'( \beta_{\tilde, t-1}^0 - \beta_{g,t-1}^0)\}_t$ and
		$\{ (x_{it}'( \beta_{\tilde g, t}^0 - \beta_{g,t}^0) -x_{i,t-1}'( \beta_{\tilde g, t-1}^0 - \beta_{g,t-1}^0)) \Delta \epsilon_{it}\}_t$ are strongly mixing process with mixing coefficients $\alpha [t]$.
		Moreover, $E (x_{it} \Delta \epsilon_{it})=0$ and $E(x_{i,t-1} \Delta \epsilon_{it}) =0$.

		\item
		\label{as-fe-xb-tail}
		There exist constants $b_x>0$, $b_e >0$, $d_{2x} >0$ and $d_{2e}$ such that $\Pr ( | x_{it}'( \beta_{\tilde g, t}^0 - \beta_{g,t}^0) - x_{i,t-1}'( \beta_{\tilde g, t-1}^0 - \beta_{g,t-1}^0) | > m ) \le \exp (1- (m/b_x))^{d_{2x}}$ and
		$\Pr ( | (x_{it}'( \beta_{\tilde g, t}^0 - \beta_{g,t}^0) - x_{i,t-1}'( \beta_{\tilde g, t-1}^0 - \beta_{g,t-1}^0)) \Delta  \epsilon_{it} | > m ) \le \exp (1- (m/b_e))^{d_{2e}}$, for any $i$, $t$ and $m >0$.
		
	\end{enumerate}
\end{assumption}

\begin{assumption}\ \\
	\label{a: ad-qs}\vspace*{-0.5cm}
	\begin{enumerate}
		\item $\{ (x_{i1}, \epsilon_{i1}), \dots, (x_{iT}, \epsilon_{iT}) \} $'s are independent over $i$.
		\item $\max_{1 \le i \le N} \max_{1 \le t \le T} E( \left\| x_{it} \right\|^{2\tau_0} ) < C < \infty $ and $\max_{1 \le i \le N} \max_{1 \le t \le T} E( \left\| \epsilon_{it} \right\|^{2\tau_0} ) < C < \infty $ for some $C$ and $\tau_0 \ge 2 $.
		\item For $\tau_0$ that satisfies Assumption~\ref{a: ad-qs}.2, there exists $\varepsilon_0 >0$ such that $N^{1-\tau_0} T (\ln T)^{\epsilon_0 \tau_0 } \to 0 $ as $N,T \to \infty$. 
	\end{enumerate}
\end{assumption}

	Under these new sets of assumptions\footnote{
	Assumption~\ref{a: ad-qs} can be implied by Assumptions~\ref{a: break} and \ref{as-jmin-imin}. Nevertheless, we introduce this assumption to directly apply the results of \citet{qian&su2016} for AGFL in the presence of individual fixed effects.}, we show that the GAGFL estimator applied to first-differenced models has asymptotic properties similar to those in the case of level models without individual-specific intercepts. As above, we first show that the estimation error in the group structure has a limited impact on the estimation of the coefficients.
Let
\begin{align*}
\mathring \beta
= \arg\min_{\beta \in \mathcal{B}^{GT}}
\left( \frac{1}{N T}
\sum_{i=1}^N \sum_{t=2}^T
(\Delta y_{it}  - x_{it} ' \beta_{g_i^0 ,t} + x_{i,t-1} ' \beta_{g_i^0, t-1})^2
+ \lambda  \sum_{g\in \mathbb{G}}
\sum_{t=2}^T \dot w_{g,t} \left\Vert \beta_{g, t} - \beta_{g, t-1}
\right\Vert \right) .
\end{align*}

\begin{lemma}
	\label{lem-fe-db-cb-h}
	Suppose that Assumptions \ref{a: fe-basic}, \ref{a: fe-iden},
	and \ref{a: fe-tail}
	are satisfied. As $N,T \to \infty$, for any $\delta >0$,
	it holds that
	$
	\hat \beta_{g,t} = \mathring \beta_{g,t} + o_p (T^{-\delta}),
	$
	for all $g$ and $t$.
\end{lemma}
Then, we can state exactly the same theorems as Theorems~\ref{lem-break-c-agfl} and~\ref{lem-agafl-date}, but under Assumptions \ref{a: break}, \ref{a: fe-basic}, \ref{a: fe-iden},  
\ref{a: fe-tail}, and \ref{a: ad-qs}, regarding the break detection.
Finally, under the correct estimation of break and group memberships, we can similarly show that the coefficient estimates of GAGFL are asymptotically equivalent to the least squares estimates under the true breakpoints and group memberships. This asymptotic equivalence and the proofs of theorems in this section are in the supplement.

\section{Extensions}
	\label{sec-extension}

		There are various directions in which to extend our proposed method. We present those extensions that are relevant in practice and analyzed relatively easily. In particular, we consider models in which a subset of coefficients are fully time varying and models in which a subset of coefficients are homogeneous and/or time invariant. Note that the two extensions considered here are indeed special cases of our model. However, our estimation method needs modification to incorporate their special features.

	\paragraph{Models with fully time-varying coefficients and time fixed effects.}

	Our model has variants with fully time-varying coefficients in special cases, but the estimation needs modification. Without loss of generality, assume that the first $k_1$ explanatory variables $x_{1,it}$ have fully time-varying coefficients, $\beta_{1, g}$, and the remaining $k-k_1$ covariates $x_{2,it}$ correspond to coefficients, $\beta_{2, g, t}$, which may exhibit structural breaks. Then, we can estimate the model by minimizing the following adjusted objective function
	\begin{align*}
	\argmin_{(\beta, \gamma) \in \mathcal{B}^{GT} \times \mathbb{G}^N}
	\frac{1}{N T} \sum_{i=1}^N \sum_{t=1}^T
	(y_{it}  - x_{1,it} ' \beta_{1, g_i, t}-x_{2,it}'\beta_{2, g_i, t} )^2
	+ \lambda \sum_{g \in \mathbb{G}} \sum_{t=2}^T
	\dot w_{g,t} \left\Vert \beta_{2, g, t} - \beta_{2, g, t-1}
	\right\Vert.\nonumber
	\end{align*}
	The resulting coefficient estimates $\hat\beta_{g,t}$ satisfy~\eqref{eq:lemma1} in Lemma~\ref{lem-db-cb-h} under the same set of assumptions. Similar theorems such as Theorems~\ref{lem-break-c-agfl} and~\ref{lem-agafl-date} hold for the $k-k_1$ subset of the estimated coefficients $\hat{\beta}_{2,g,t}$ with slightly modified assumptions. A complete theoretical analysis of this algorithm is in the supplement, in which we also apply it to an application. 
	An important special case is a model with group-specific time fixed effects, i.e.
	\begin{align*}
	y_{it} = \lambda_{g_i,t} + z_{it}'\delta_{g_i,t} + \epsilon_{it},
	\end{align*}
	where $\lambda_{g,t}$ is a group-specific time effect for group $g$ and period $t$, $z_{it}$ is the explanatory variables, and $\delta_{g_i,t}$ is the associated coefficients characterized by structural breaks.
	We assume that $\lambda_{g,t}$ changes at every period, and we do not penalize a change in $\lambda_{g,t}$.

	\paragraph{Models with partially homogeneous and/or time-invariant coefficients.}
		
	Our model can also incorporate the case of partially time-invariant coefficients. Without loss of generality, assume that the first $k_1$ explanatory variables $x_{1,it}$ have time-invariant coefficients, $\beta_{1, g}$, and the remaining $k-k_1$ covariates $x_{2,it}$ correspond to time-varying coefficients, $\beta_{2, g, t}$. Then, we can estimate the model by minimizing the following adjusted objective function
	\begin{align*}
	\argmin_{(\beta, \gamma) \in \mathcal{B}^{GT} \times \mathbb{G}^N}
	\frac{1}{N T} \sum_{i=1}^N \sum_{t=1}^T
	(y_{it}  - x_{1,it} ' \beta_{1, g_i}-x_{2,it}'\beta_{2, g_i, t} )^2
	+ \lambda \sum_{g \in \mathbb{G}} \sum_{t=2}^T
	\dot w_{g,t} \left\Vert \beta_{2, g, t} - \beta_{2, g, t-1}
	\right\Vert.\nonumber
	\end{align*}
	The iterative algorithm remains the same except that the AGFL only penalizes a part of the coefficient vector $\beta_{2,g,t}$. 

	If the coefficients of $x_{1,it}$ are not only time invariant but also homogeneous, then the objective function in this case becomes
	\begin{align*}
	\argmin_{(\beta, \gamma) \in \mathcal{B}^{GT} \times \mathbb{G}^N}
	\frac{1}{N T} \sum_{i=1}^N \sum_{t=1}^T
	(y_{it}  - x_{1,it} ' \beta_{1}-x_{2,it}'\beta_{2, g_i, t} )^2
	+ \lambda \sum_{g \in \mathbb{G}} \sum_{t=2}^T
	\dot w_{g,t} \left\Vert \beta_{2, g, t} - \beta_{2, g, t-1}
	\right\Vert.\nonumber
	\end{align*}
	To find a minimizer for this objective function, we adjust the iterative algorithm by separating the estimation of $\beta_1$ and $\beta_{2,g,t}$ such that a penalized minimization is applied only to $\beta_{2,g,t}$ and an estimate of $\beta_1$ is updated using all the observations  rather than observations in a group.

	\section{Monte Carlo simulation}\label{sec:simulation}
	This section evaluates the finite sample performance of the proposed GAGFL estimator.
	In particular, we investigate if GAGFL can correctly classify units
	and whether it can effectively detect structural breaks.
	We also examine the importance of taking heterogeneity of structural breaks into account.

	\subsection{Data generation process}\label{sec:bh-dgp}
	
	We study four data generation processes differing in the distribution of errors, the specification of fixed effects, and the inclusion of lagged dependent variables.
	\begin{itemize}
		\item[] \textbf{[DGP.1]} We generate the data from the following regression model:
		\begin{equation}
		y_{it}= x_{it}\beta_{g_i,t}+\epsilon_{it},\nonumber
		\end{equation}
		where
		$x_{it}\sim \textrm{ i.i.d. } N(0,1)$ and
		$\epsilon_{it}\sim \textrm{ i.i.d. } N(0,\sigma_\epsilon^2)$.
		We consider $\sigma_\epsilon=(0.5,0.75)$.
		There are three groups.
		Let $N_j$, $j= 1,2,3$, denote the number of units in group $j$.
		Note that $N= N_1 + N_2 + N_3$.
		We fix the ratio of units among groups such that $N_1:N_2:N_3=0.3:0.3:0.4$.
		The coefficients in the three groups are given, respectively, by
		\begin{align*}
		\beta_{1,t}=\left\{\begin{array}{ll}
		1 & \text{if } 1\leq t < \lfloor T/2\rfloor\\
		2 & \text{if } \lfloor T/2\rfloor\leq t < \lfloor 5T/6\rfloor\\
		3 & \text{if } \lfloor 5T/6\rfloor\leq t \leq T\\
		\end{array}\right.,
		\quad
		\beta_{2,t}=\left\{\begin{array}{ll}
		3 & \text{if } 1\leq t < \lfloor T/3\rfloor\\
		4 & \text{if } \lfloor T/3\rfloor\leq t < \lfloor 5T/6\rfloor\\
		5 & \text{if } \lfloor 5T/6\rfloor\leq t \leq T\\
		\end{array}\right.,
		\end{align*}
		where $\lfloor\cdot\rfloor$ takes the integer part, and
		$
		\beta_{3,t}=1.5 \quad \textrm{for all } 1\leq t\leq T$.
		
		The first group exhibits two structural breaks. Two breaks also characterize units in the second group, but they occur at different time points. For the third group, the slope coefficient is stable without a break. 
		
		\item[] \textbf{[DGP.2]} Same as DGP.1 except that $\epsilon_{it}$ follows an AR(1) process for each individual $i$: $\epsilon_{it}=	0.5\epsilon_{i,t-1} + u_{it}$, where $u_{it}\sim \textrm{ i.i.d. } N(0, 0.75)$.
		
		\item[] \textbf{[DGP.3]} Same as DGP.1 except that there is an additive individual fixed effect, namely
		$y_{it}= \mu_i+ x_{it}\beta_{g_i,t}+\epsilon_{it}$, where $\mu_i = T^{-1}\sum_{t=1}^T x_{it}$.
		
		\item[] \textbf{[DGP.4]} Same as DGP.1 except that a lagged dependent variable is included, namely
		$$
		y_{it}= y_{it-1}\tau_{g_i,t}+x_{it}\beta_{g_i,t}+\epsilon_{it},
		$$		
		where the coefficient of the lagged dependent variable is given by
		\begin{align*}
		\tau_{1,t}=\left\{\begin{array}{ll}
		0.2 & \text{if } 1\leq t < \lfloor T/2\rfloor\\
		0.8 & \text{if } \lfloor T/2\rfloor\leq t < \lfloor 5T/6\rfloor\\
		0.2 & \text{if } \lfloor 5T/6\rfloor\leq t \leq T\\
		\end{array}\right.,
		\quad
		\tau_{2,t}=\left\{\begin{array}{ll}
		-0.3 & \text{if } 1\leq t < \lfloor T/3\rfloor\\
		-0.6 & \text{if } \lfloor T/3\rfloor\leq t < \lfloor 5T/6\rfloor\\
		-0.9 & \text{if } \lfloor 5T/6\rfloor\leq t \leq T\\
		\end{array}\right.,
		\end{align*}
		and
		$\tau_{3,t}=0.5 \quad \textrm{for all } 1\leq t\leq T$.

	\end{itemize}
	
	DGP.1 is our benchmark case. DGP.2 exhibits serial correlation in the errors. DGP.3 includes individual fixed effects, and thus first-differenced data are used as in~\eqref{eq:first-differece}. DGP.4 follows a dynamic panel model. For each DGP, we consider two cross-sectional sample sizes, $N=(50,100)$, and three lengths of time series,
	$T=(10, 20, 40)$.
	In total, we have six combinations of cross-sectional sample size and length of time series.

	\subsection{Procedures}
	
	We estimate the model by GAGFL.
	We also examine the performance of the penalized least squares (PLS) by \citet{qian&su2016} ignoring heterogeneity, and that of \emph{common} break detection in heterogeneous panels by \citet{baltagi&qu&kao2016} (hereafter BFK). Note that PLS is equivalent to GAGFL with $G=1$.

	For GAGFL, the number of groups is chosen by minimizing the BIC defined in Section~\ref{sec-choose-tuning}.
	For both GAGFL and PLS, we follow \citet{qian&su2016} for the computational method and choice of tuning parameters.
	To estimate the breaks for a given grouped pattern, we employ the block-coordinate descent algorithm for solving the PLS. The tuning parameter $\lambda$ is selected by minimizing the information criterion in the interval of $[0.01,100]$, where the upper bound leads to breaks in all time points while the lower bound leads to no breaks.
	To construct the weights $\{\dot{\omega}_{g,t}\}$, we set $\kappa=2$ following the adaptive Lasso literature.
	In DGP.3, first differencing is employed to eliminate individual fixed effects for GAGFL and PLS (see Section \ref{sec-extension}).
		
	BFK detects breaks by minimizing the sum of squared residuals over distinct breakpoints, where the residuals are from the individual time series estimation. The method is applied to detect multiple \emph{common} breaks occurring for all individual units, although the slope coefficients are allowed to be individual specific. To implement this method, we need to specify the number of breaks, and we set the number of breaks equal to three, the true number for the \emph{pooled} data.\footnote{Groups 1 and 2 share a break at the same time $\lfloor5T/6\rfloor$, and each of them experiences another break at different times $\lfloor T/2\rfloor$ and $\lfloor T/3\rfloor$, respectively. Thus, there are in total three breaks for the pooled data.} We employ the  approach discussed by \citet{baltagi&qu&kao2016} (see also \citet{bai97et,bai2010}) that estimates multiple breakpoints sequentially. As the BFK method estimates individual time series separately, no transformation is needed under DGP.3.
		In this case, breaks are only allowed in the slope coefficient $\beta$, but not in the intercept.
	
	\subsection{Evaluation criteria}

	We evaluate the performance of the proposed method from five different perspectives:
	selecting the right number of groups, clustering, determining the number of breaks, the break date estimates, and the coefficient estimates.
	
	First, we evaluate the performance of the BIC in determining the number of groups by computing the empirical probability of selecting a particular number.
	Second, we measure the clustering accuracy by the average of the misclassification frequency ($\widehat{g}_i\neq g^0_i$) across replications. Let $I(\cdot)$ be the indicator function. The misclassification frequency (MF) is the ratio of misclassified units to the total number of units, i.e.
	$
	\textrm{MF} = 1/N\sum_{i=1}^N I(\widehat{g}_i\neq g^0_i).
	$
	
	The third and fourth criteria concern break estimation.
	The third criterion is the average frequency of correctly estimating the number of breaks. 
	The fourth criterion is the average Hausdorff error of break date estimates. This measure is also used by \citet{qian&su2016}. 
	The Hausdorff error is the Hausdorff distance (HD) between the estimated break dates and the true set of dates, i.e.
	$$
	\textrm{HD}(\widehat{T}_{g,\widehat{m}}^0,T_{g,m^0}^0)\equiv\max\{\mathcal{D}(\widehat{T}_{g,\widehat{m}}^0,T_{g,m^0}^0), \mathcal{D}(T_{g,m^0}^0,\widehat{T}_{g,\widehat{m}}^0)\},
	$$
	where $\mathcal{D}(A,B)\equiv\sup_{b\in B}\inf_{a\in A}|a-b|$ for any set $A$ and $B$. We report the Hausdorff error multiplied by 100 and divided by $T$ for each group, i.e. $100\times \textrm{HD}(\widehat{T}_{g,\widehat{m}}^0,T_{g,m^0}^0)/T$, averaged across the replications.
	
	Lastly, we evaluate the accuracy of the coefficient estimates using their root mean squared error (RMSE) and the coverage probability of the two-sided nominal 95\% confidence interval. We compute the overall root mean square error for all units at each time as
	$$
	\textrm{RMSE}(\widehat{\beta}_{it}) = \sqrt{\frac{1}{NT}\sum_{i=1}^N\sum_{t=1}^T(\widehat{\beta}_{it}-\beta_{it})^2}. $$
	The coverage probability is computed as
	$$
	\textrm{Coverage}(\widehat{\beta}_{it}) = \frac{1}{NT}\sum_{i=1}^N\sum_{t=1}^TI(\widehat{\beta}_{it}-1.96\widehat{\sigma}_{\beta,it}\leq \beta_{it}\leq\widehat{\beta}_{it}+1.96\widehat{\sigma}_{\beta,it}),
	$$
	where $\widehat{\sigma}_{\beta,it}$ is the
	estimated standard deviation of $\widehat{\beta}_{it}$. We average all evaluation measures across 1,000 replications.
	
	\subsection{Determining the number of groups}
	As the implementation of our method requires the specification of the number of groups, we first examine how well the BIC proposed in Section~\ref{sec-choose-tuning} performs in selecting this number. To implement the BIC, we choose the scaling parameter $\hat{\sigma}^2$ to be the sum of squared residuals obtained by plugging in the estimates from the homogeneous panel, i.e. $G=1$.\footnote{While this may not yield a consistent estimate of the variance of residuals, the main role of $\hat{\sigma}^2$ is to scale the penalty term, such that it is invariant to the variation of the data. In fact, we find that choosing $\hat{\sigma}^2$ as the sum of squared residuals obtained under $G_{max}$ groups can lead to rather unstable results that vary across the specifications of $G_{max}$.}

\begin{table}[htp]
	\begin{center}\caption{Group number selection frequency using BIC when $G_0=3$}\label{tab:ic}\small
		\begin{tabular}{cccccccccccccccccc}
			\hline\hline
			$ \sigma_\epsilon$ & $N$ & $T$ & 1 & 2 & 3 & 4 & 5 && & 1 & 2 & 3 & 4 & 5\\
			\hline
			&     &       & \multicolumn{5}{c}{DGP.1} && & \multicolumn{5}{c}{DGP.2}\\
			\cline{4-8} \cline{11-15}
			0.5      & 50  & 10    & 0.000	& 0.002	& \textbf{0.996}& 0.002	&0.000 &&							& 0.000	& 0.000	& \textbf{1.000}& 0.000	&0.000\\
			&50   & 20    & 0.000	& 0.000	& \textbf{0.988}& 0.012	&0.000&&						    & 0.000	& 0.000	& \textbf{0.996}& 0.004	&0.000\\
			&50   & 40    & 0.000	& 0.000 & \textbf{0.987}& 0.011	&0.002&&  						    & 0.000	& 0.000	& \textbf{0.996}& 0.002	&0.002\\
			&100  & 10    & 0.000	& 0.000	& \textbf{0.998}& 0.000	&0.002&&						    & 0.000	& 0.000	& \textbf{0.998}& 0.002 &0.000\\	
			&100  & 20    & 0.000  & 0.000 & \textbf{0.994}& 0.006 &0.000&&  						    & 0.000	& 0.000	& \textbf{0.993}& 0.007	&0.000\\
			&100  & 40    & 0.000  & 0.000 & \textbf{0.992}& 0.008 &0.000&&         					& 0.000	& 0.000	& \textbf{0.996}& 0.004	&0.000\\
			\\
			0.75     & 50  & 10   & 0.000	& 0.020	& \textbf{0.974}& 0.004	&0.002&&							& 0.000	& 0.007	& \textbf{0.990}& 0.000	&0.003\\
			&50   & 20   & 0.000	& 0.006	& \textbf{0.982}& 0.012	&0.000&&     						& 0.000	& 0.004	& \textbf{0.995}& 0.001	&0.000\\
			&50   & 40   & 0.000	& 0.002 & \textbf{0.988}& 0.010	&0.000&&   						    & 0.000	& 0.000	& \textbf{0.998}& 0.002	&0.000\\
			&100  & 10   & 0.000	& 0.000	& \textbf{0.998}& 0.002	&0.000&&  							& 0.000	& 0.000	& \textbf{0.997}& 0.003 &0.000\\	
			&100  & 20   & 0.000   & 0.000 & \textbf{0.997}& 0.003 &0.000&&   							& 0.000	& 0.000	& \textbf{0.998}& 0.002	&0.000\\
			&100  & 40   & 0.000   & 0.000 & \textbf{0.984}& 0.014 &0.002&&   							& 0.000	& 0.000	& \textbf{0.994}& 0.006 &0.000\\
			\hline
			&     &       & \multicolumn{5}{c}{DGP.3} && & \multicolumn{5}{c}{DGP.4}\\
			\cline{4-8} \cline{11-15}                                                     		
			0.5      & 50  & 10    & 0.000	& 0.005	& \textbf{0.989}& 0.006	&0.000 &&							& 0.000	& 0.062	& \textbf{0.938}& 0.000	&0.000\\       		
			&50   & 20    & 0.000	& 0.000	& \textbf{0.991}& 0.009	&0.000&&						    & 0.000	& 0.030	& \textbf{0.962}& 0.008	&0.000\\     			
			&50   & 40    & 0.000	& 0.000 & \textbf{0.990}& 0.010	&0.000&&  						    & 0.000	& 0.000	& \textbf{0.996}& 0.002	&0.002\\    		
			&100  & 10    & 0.000	& 0.000	& \textbf{0.990}& 0.008	&0.002&&						    & 0.000	& 0.004	& \textbf{0.991}& 0.005 &0.000\\	    			
			&100  & 20    & 0.000  & 0.000 & \textbf{0.994}& 0.006 &0.000&&  						    & 0.000	& 0.001	& \textbf{0.999}& 0.000	&0.000\\   			
			&100  & 40    & 0.000  & 0.000 & \textbf{0.992}& 0.008 &0.000&&         					& 0.000	& 0.000	& \textbf{0.999}& 0.001	&0.000\\ 		
			\\                                                                                                                           		
			0.75     & 50  & 10   & 0.000	& 0.002	& \textbf{0.884}& 0.100	&0.014 &&							& 0.000	& 0.139	& \textbf{0.861}& 0.000	&0.000\\       		
			&50   & 20   & 0.000	& 0.000	& \textbf{0.946}& 0.048	&0.006&&     						& 0.000	& 0.041	& \textbf{0.959}& 0.000	&0.000\\         		
			&50   & 40   & 0.000	& 0.016 & \textbf{0.976}& 0.008	&0.000&&   						    & 0.000	& 0.000	& \textbf{0.998}& 0.002	&0.000\\        		
			&100  & 10   & 0.000	& 0.016	& \textbf{0.972}& 0.012	&0.000&&  							& 0.000	& 0.016	& \textbf{0.984}& 0.000 &0.000\\	         		
			&100  & 20   & 0.000   & 0.020 & \textbf{0.982}& 0.016 &0.000&&   							& 0.000	& 0.003	& \textbf{0.997}& 0.000	&0.000\\     		
			&100  & 40   & 0.000   & 0.000 & \textbf{0.986}& 0.012 &0.002&&   							& 0.000	& 0.000	& \textbf{0.995}& 0.005 &0.000\\     				
			\hline
		\end{tabular}
	\end{center}
\end{table}

	Table~\ref{tab:ic} presents the empirical probability that a particular number of groups, ranging from $G=1$ to 5, is selected according to the BIC. Recall that the true number is three.
	The results indicate that the BIC can identify the correct group size with a high probability. In DGP.1 and DGP.2, the BIC selected the correct number of groups in more than 97\% of the cases, even when $T=10$ and $\sigma_\epsilon=0.75$. In DGP.3, the correct selection rate is also high when $\sigma_\epsilon=0.5$, although it is relatively low when $\sigma_\epsilon=0.75$ and the sample sizes are small, possibly because of first differencing. Introducing lagged dependent variables deteriorates performance slightly, although the correct selection rate still exceeds 86\%, even in the worst case. 
	In Section S.1.1 of the supplement, we examine the BIC constructed from initial estimates. This continues to work well but generally performs more poorly than that defined in~\eqref{eq:BIC} using the final estimates.

	\subsection{Clustering, break detection, and point estimation}
	
	The previous section suggests that the information criteria are useful in determining the number of groups. Given the true number of groups, we now examine the accuracy of clustering, break detection, and coefficient estimation.
	
	\subsubsection*{Clustering accuracy}
	We first investigate the clustering accuracy of the proposed estimator. Table~\ref{tab:clustering} presents the averaged misclustering frequencies.
	This shows that the proposed method can correctly classify a large proportion of units. The misclassification frequency generally reduces with $T$ but not with $N$.
	In DGP.1 and DGP.2, GAGFL can correctly classify more than 95\% of individuals even with a relatively short time series
	($T=10$) and large errors ($\sigma_\epsilon=0.75$).
	Allowing individual fixed effects in DGP.3 leads to slightly less accurate clustering. Nevertheless, GAGFL still manages to limit the misclassification frequency to less than 7\% for a short $T$, and the rate drops quickly as $T$ increases. Introducing dynamic effects as in DGP.4 scarcely affects the clustering, and more than 96\% of individuals can be correctly classified.
	These results demonstrate that GAGFL can effectively capture grouped patterns of heterogeneity.

	\begin{table}[t]
		\begin{center}
			\caption{Average misclassification frequency}\label{tab:clustering}\small
			\begin{tabular}{lllcccccccccc}
				\hline\hline
				&& \multicolumn{3}{c}{$N=50$} && \multicolumn{3}{c}{$N=100$}\\
				&& $T=10$ & $T=20$  & $T=40$ && $T=10$ & $T=20$ & $T=40$\\
				\hline
				DGP.1                        &$\sigma_\epsilon=0.5$         &0.0104   & 0.0026  & 0.0010&& 0.0097 & 0.0015 & 0.0000\\
				&$\sigma_\epsilon=0.75$        &0.0448   & 0.0177  & 0.0027&& 0.0377 & 0.0140 & 0.0042\\
				\\
				DGP.2                        &$\sigma_\epsilon=0.5$         &0.0048   & 0.0025  & 0.0010&& 0.0040 & 0.0022 & 0.0001\\
				&$\sigma_\epsilon=0.75$        &0.0296   & 0.0076  & 0.0028&& 0.0206 & 0.0081 & 0.0042\\
				\\
				DGP.3                        &$\sigma_\epsilon=0.5$         &0.0179   & 0.0024  & 0.0001&& 0.0171 & 0.0028 & 0.0020\\
				&$\sigma_\epsilon=0.75$        &0.0663   & 0.0240  & 0.0041&& 0.0484 & 0.0161 & 0.0013\\
				\\
				DGP.4                        &$\sigma_\epsilon=0.5$         &0.0074   & 0.0005  & 0.0024&& 0.0059 & 0.0004 & 0.0040\\
				&$\sigma_\epsilon=0.75$        &0.0357   & 0.0114  & 0.0016&& 0.0327 & 0.0046 & 0.0002\\			
				\hline
			\end{tabular}
		\end{center}
	\end{table}

	\subsubsection*{Break estimation accuracy}
	
	We now examine the accuracy of structural break detection.
	We compare the estimated number of breaks for GAGFL and PLS, and compare the accuracy of the estimated breakpoints for GAGFL, PLS, and BFK. Recall that the latter two approaches both assume common breaks to all units at the same time, although BFK allows individually heterogeneous coefficients.

	\begin{table}[htp]
		\begin{center}
			\caption{Average frequency of correct estimation of the number of breaks}\label{tab:breaknumberfreq}\small
			\begin{tabular}{lllcccccccccc}
				\hline\hline
				&&                     & \multicolumn{3}{c}{$N=50$} && \multicolumn{3}{c}{$N=100$}\\
				&\multicolumn{2}{c}{Group (True break num.)}  & $T=10$ & $T=20$ & $T=40$ && $T=10$ & $T=20$ & $T=40$\\
				\hline
				&&                     &\multicolumn{5}{c}{DGP.1}\\
				$\sigma_\epsilon=0.5$ & GAGFL    &G1 ($m_{1,0}^0=2$)               &0.996  & 0.994   & 0.996  && 0.995 & 1.000 & 0.997 \\
				&&G2 ($m_{2,0}^0=2$)               &0.997  & 0.997   & 1.000  && 1.000 & 1.000 & 0.991\\
				&&G3 ($m_{3,0}^0=0$)               &0.996  & 0.999   & 0.999  && 0.985 & 0.994 & 1.000\\
				& PLS        & All individuals                        &0.310  & 0.263   & 0.201  && 0.407 & 0.386 & 0.223\\
				\\
				$\sigma_\epsilon=0.75$ & GAGFL &G1 ($m_{1,0}^0=2$)                 &0.786  & 0.867   & 0.962 && 0.962 & 0.991  & 0.994\\
				&&G2 ($m_{2,0}^0=2$)              &0.852  & 0.914   & 0.978 && 0.986 & 0.997  & 0.996\\
				&&G3 ($m_{3,0}^0=0$)              &0.812  & 0.956   & 0.994 && 0.964 & 0.993  & 0.998\\
				& PLS       & All individuals                         &0.268  & 0.248   & 0.164 && 0.389 & 0.361  & 0.196\\
				\hline
				&&                     &\multicolumn{5}{c}{DGP.2}\\
				$\sigma_\epsilon=0.5$ & GAGFL    &G1 ($m_{1,0}^0=2$)               &0.997  & 0.996   &  0.996&& 0.999 & 1.000  & 0.999\\
				&&G2 ($m_{2,0}^0=2$)               &1.000  & 1.000   &  1.000&& 1.000 & 1.000  & 0.999\\
				&&G3 ($m_{3,0}^0=0$)               &0.998  & 0.998   &  0.999&& 0.995 & 0.994  & 0.997\\
				& PLS        & All individuals                        &0.308  & 0.260   &  0.201&& 0.416 & 0.386  & 0.249\\
				\\
				$\sigma_\epsilon=0.75$ & GAGFL &G1 ($m_{1,0}^0=2$)                &0.926  & 0.976   &  0.979&& 0.992 & 1.000  & 0.998\\
				&&G2 ($m_{2,0}^0=2$)              &0.957  & 0.973   &  0.983&& 0.995 & 1.000  & 1.000\\
				&&G3 ($m_{3,0}^0=0$)              &0.929  & 0.986   &  0.994&& 0.987 & 0.995  & 0.998\\
				& PLS       & All individuals                         &0.295  & 0.242   &  0.193&& 0.408 & 0.410  & 0.240\\
				\hline
				
				&&                     &\multicolumn{5}{c}{DGP.3}\\
				$\sigma_\epsilon=0.5$ & GAGFL    &G1 ($m_{1,0}^0=2$)               &0.789  & 0.904   &  0.945&& 0.933 & 0.979  & 0.993\\
				&&G2 ($m_{2,0}^0=2$)               &0.793  & 0.903   &  0.941&& 0.945 & 0.991  & 0.992\\
				&&G3 ($m_{3,0}^0=0$)               &0.787  & 0.939   &  0.989&& 0.926 & 0.985  & 0.985\\
				& PLS        & All individuals                        &0.103  & 0.039   &  0.019&& 0.159 & 0.082  & 0.050\\
				\\
				$\sigma_\epsilon=0.75$ & GAGFL &G1 ($m_{1,0}^0=2$)                &0.268  & 0.362   &  0.597&& 0.539 & 0.730  & 0.873\\
				&&G2 ($m_{2,0}^0=2$)              &0.341  & 0.382   &  0.551&& 0.609 & 0.773  & 0.902\\
				&&G3 ($m_{3,0}^0=0$)              &0.227  & 0.455   &  0.812&& 0.543 & 0.772  & 0.947\\
				& PLS       & All individuals                         &0.095  & 0.022   &  0.008&& 0.129 & 0.064  & 0.033\\
				\hline
				
				&&                     &\multicolumn{5}{c}{DGP.4}\\
				$\sigma_\epsilon=0.5$ & GAGFL    &G1 ($m_{1,0}^0=2$)                &0.986  & 0.985   &  0.989&& 0.997 & 0.999  & 0.999\\
				&&G2 ($m_{2,0}^0=2$)                                                                             &0.977  & 0.982    &  0.976&& 1.000 & 0.999  & 1.000\\
				&&G3 ($m_{3,0}^0=0$)                                                                             &0.990  & 0.996   &  0.997&& 0.996 & 0.999  & 0.999\\
				& PLS        & All individuals                                                                          &0.051  & 0.016  &  0.002&& 0.077 & 0.013  & 0.020\\
				\\
				$\sigma_\epsilon=0.75$ & GAGFL &G1 ($m_{1,0}^0=2$)                 &0.830  & 0.877   &  0.873&& 0.986 & 0.990  & 0.992\\
				&&G2 ($m_{2,0}^0=2$)                                                                             &0.832  & 0.849   &  0.780&& 0.989 & 0.991  & 0.988\\
				&&G3 ($m_{3,0}^0=0$)                                                                             &0.943  & 0.988   &  0.995&& 0.993 & 0.996  & 0.998\\
				& PLS       & All individuals                                                                           &0.044  & 0.011   &  0.002&& 0.066 & 0.012  & 0.020\\
				\hline			
			\end{tabular}
		\end{center}
	\end{table}
	
	Table~\ref{tab:breaknumberfreq} presents the average frequencies of correctly estimating the number of breaks for each group in our method, and the same frequency for all units in PLS.\footnote{BFK is not compared here because the true number of breaks is assigned to estimate the breakpoints.}
	Because the panel is heterogeneous, the frequency of PLS is calculated by comparing the estimated number of breaks with the true number of breaks in the pooled data: three in our case.
	This shows that when errors are of a moderate size ($\sigma_\epsilon=0.5$), our method almost perfectly detects the correct number of breaks (with a correct detection frequency of more than 97\%) except in DGP.3.
	With individual fixed effects in DGP.3, the frequency is about 78\% in the small sample with $N=50$ and $T=10$, but this frequency quickly improves to more than 90\% when $T=20$ or when $N=100$.
	In the case with relatively large errors ($\sigma_\epsilon=0.75$),
	GAGFL still works well, although less accurately than in the case of $\sigma_\epsilon=0.5$.
	When the sample size is small ($N=50$ and $T=10$),
	it can correctly estimate the number of breaks in at least 78\% of the cases in DGP.1, 92\% in DGP.2, 22\% in DGP.3, and 83\% in DGP.4. Unreported results suggested that when GAGFL fails to detect the correct number of breaks, it typically overestimates them. Again, the correct detection frequency increases quickly with $N$ and $T$. For example, it quickly reaches 90\% on average in DGP.3 when $N=100$ and $T=40$. In contrast, the correct detection frequency remains less than 40\% for PLS in DGP.1 and DGP.2, and even lower in DGP.3 and DPG.4. Moreover, the frequency does not seem to increase with the sample size. 

	\begin{table}[htp]
		\begin{center}
			\caption{Hausdorff error of break date estimates}\label{tab:hd-error}\small
			\begin{tabular}{lllcccccccccc}
				\hline\hline
				&&                     & \multicolumn{3}{c}{$N=50$} && \multicolumn{3}{c}{$N=100$}\\
				&\multicolumn{2}{c}{Group (True break num.)}  & $T=10$ & $T=20$ & $T=40$ && $T=10$ & $T=20$ & $T=40$\\
				\hline
				&&                     &\multicolumn{5}{c}{DGP.1}\\
				$\sigma_\epsilon=0.5$ & GAGFL    &G1 ($m_{1,0}^0=2$)               &0.0015  & 0.0014  & 0.0014  && 0.0007 & 0.0008  & 0.0023\\
				&&G2 ($m_{2,0}^0=2$)                                                                           &0.0013  & 0.0009  & 0.0006  && 0.0006 & 0.0008  & 0.0009\\
				& BFK                                & All individuals                                           &0.1445  & 0.1695  & 0.1895  && 0.1333 & 0.1351  & 0.1738\\
				& PLS                                & All individuals                                           &0.1682  & 0.1357  & 0.1318   && 0.1628 & 0.1182  & 0.1056\\
				\\
				$\sigma_\epsilon=0.75$ & GAGFL &G1 ($m_{1,0}^0=2$)               & 0.0399 & 0.0216   & 0.0061  &&   0.0057 & 0.0018  & 0.0011\\
				&&G2 ($m_{2,0}^0=2$)                                                                          & 0.0231 & 0.0137   & 0.0035  &&   0.0018 & 0.0007  & 0.0006\\
				& BFK       & All individuals                                                                   & 0.1835 & 0.2683   & 0.2447  &&   0.1585 & 0.2452  & 0.2254\\
				& PLS       & All individuals                                                                    & 0.1670 & 0.1418   & 0.1349  &&   0.1621 &  0.1211 & 0.1087\\
				\hline
				&&                     &\multicolumn{5}{c}{DGP.2}\\
				$\sigma_\epsilon=0.5$ & GAGFL    &G1 ($m_{1,0}^0=2$)               &0.0003  & 0.0004  & 0.0021   && 0.0004 & 0.0005 & 0.0010\\
				&&G2 ($m_{2,0}^0=2$)                                                                           &0.0002  & 0.0003  & 0.0004  && 0.0000 & 0.0003 & 0.0005\\
				& BFK       & All individuals                                                                     &0.1406  & 0.1320  & 0.1399  && 0.1370 & 0.1129 & 0.1171\\
				& PLS                                & All individuals                                            &0.1707  & 0.1372  & 0.1308  && 0.1594 & 0.1212 & 0.1035\\
				\\
				$\sigma_\epsilon=0.75$ & GAGFL &G1 ($m_{1,0}^0=2$)               &0.0135  & 0.0052  & 0.0050  && 0.0029 & 0.0005  & 0.0006\\
				&&G2 ($m_{2,0}^0=2$)                                                                          &0.0085  & 0.0037  & 0.0034  && 0.0008 & 0.0005  & 0.0006\\
				& BFK       & All individuals                                                                    &0.1825  & 0.2019  & 0.2216  && 0.1713 & 0.1610  & 0.1985\\
				& PLS       & All individuals                                                                    &0.1700  & 0.1407  & 0.1338  && 0.1613 & 0.1220  & 0.1067\\
				\hline
				
				&&                     &\multicolumn{5}{c}{DGP.3}\\
				$\sigma_\epsilon=0.5$ & GAGFL    &G1 ($m_{1,0}^0=2$)               &0.0306  & 0.0148  & 0.0077   && 0.0089 & 0.0032 & 0.0032\\
				&&G2 ($m_{2,0}^0=2$)                                                                           &0.0309  & 0.0116  & 0.0080  && 0.0078 & 0.0010 & 0.0028\\
				& BFK       & All individuals                                                                    &0.1344  & 0.2283  & 0.2035  && 0.1211 & 0.1772 & 0.1965\\
				& PLS                                & All individuals                                           &0.1632  & 0.1604  & 0.1682  && 0.1536 & 0.1368 & 0.1316\\
				\\
				$\sigma_\epsilon=0.75$ & GAGFL &G1 ($m_{1,0}^0=2$)               &0.1101  & 0.1016  & 0.0631  && 0.0622 & 0.0447  & 0.0179\\
				&&G2 ($m_{2,0}^0=2$)                                                                          &0.0794  & 0.0880  & 0.0651  && 0.0483 & 0.0301  & 0.0114\\
				& BFK       & All individuals                                                                    &0.1817  & 0.3442  & 0.2552  && 0.1657 & 0.3443  & 0.2348\\
				& PLS       & All individuals                                                                     &0.1625  & 0.1649  & 0.1776  && 0.1509 & 0.1445  & 0.1379\\
				\hline
				
				&&                     &\multicolumn{5}{c}{DGP.4}\\
				$\sigma_\epsilon=0.5$ & GAGFL    &G1 ($m_{1,0}^0=2$)               &0.0033  & 0.0028  & 0.0013  && 0.0006 & 0.0005 & 0.0004\\
				&&G2 ($m_{2,0}^0=2$)                                                                           &0.0036  & 0.0029  & 0.0030  && 0.0002 & 0.0001 & 0.0002\\
				& BFK       & All individuals                                                                        &0.1176  & 0.1678  & 0.3437    && 0.1517 & 0.1833  & 0.3255\\
				& PLS                                & All individuals                                               &0.1782  & 0.1630  &  0.1661  && 0.1681 & 0.1511 & 0.1507 \\
				\\
				$\sigma_\epsilon=0.75$ & GAGFL &G1 ($m_{1,0}^0=2$)               &0.0315  & 0.0186  & 0.0192  && 0.0028 & 0.0012  & 0.0014\\
				&&G2 ($m_{2,0}^0=2$)                                                                          &0.0287  & 0.0244 & 0.0349  && 0.0019 & 0.0019  & 0.0015\\
				& BFK       & All individuals                                                                        &0.1417  & 0.1538  & 0.4750  && 0.1539 & 0.1704  & 0.4750\\
				& PLS       & All individuals                                                                        &0.1778  & 0.1684  & 0.1702  && 0.1687 & 0.1534  & 0.1570 \\
				\hline		
			\end{tabular}
		\end{center}
		\footnotesize{\emph{Notes:} HD ratios of GAGFL estimates for G3 (with no breaks) not reported because all zero.}
	\end{table}

	As another measure of break estimation accuracy in Table~\ref{tab:hd-error} we report the Hausdorff errors between the true break dates and those estimated by GAGFL, PLS, and BFK, conditional on the correction estimation of the number of breaks. 
	We can see that the Hausdorff errors of GAGFL are much smaller than those of PLS and BFK in all cases.
Although BFK allows for heterogeneous coefficients, it has difficulty detecting breaks in small samples because heterogeneous breaks become less visible by treating them as common, which is also noted by \citet{baltagi&qu&kao2016}.
	These results jointly illustrate the importance of accounting for heterogeneity in breaks, and show that our method can detect the number of breaks correctly and identify the breakpoints precisely, even when the error variance is large.

	\subsubsection*{Coefficient estimation accuracy}
	
	Finally, we compare the accuracy of the coefficient estimates obtained from GAGFL, PLS, and BFK. Table~\ref{tab:coef} presents the average RMSE and coverage probability of the coefficient estimates of the three methods across 1,000 replications.\footnote{To compute the average statistics, we map the estimated coefficients in a regime-group pair to each $(i,t)$ observation based on the estimated group memberships and break dates, and then compute the average RMSE and coverage probability of the entire coefficient vector. These average statistics facilitate comparison with the estimates produced by the two competing methods, PLS and BFK, because they do not produce group-specific estimates. Moreover, the average statistics also reflect the accuracy of membership and break date estimation.}
	To conserve space,
	for DGP.4 we report only the results for the coefficient of the lagged dependent variable.
	In general, we can see that the RMSE of GAGFL is much smaller than that of PLS and BFK. 
	Increasing the sample size reduces the RMSE and improves the coverage probability of GAGFL. In contrast, the RMSE of PLS does not improve as the sample size increases,
	and its coverage probability remains low. BFK has better coverage probabilities than PLS (although still lower than GAGFL), but at the cost of a much larger RMSE caused by its particularly large standard deviation. This is because BFK uses individual time series estimation, which can be rather inefficient in finite samples.

\begin{table}[htp]
	\begin{center}
		\caption{Root mean squared error and coverage probability of coefficient estimates}\label{tab:coef}
		\begin{tabular}{ccccccccccc}\hline\hline
			&       &       &       & \multicolumn{3}{c}{RMSE} &       & \multicolumn{3}{c}{Coverage} \\
			\cmidrule{5-7}\cmidrule{9-11}          & \multicolumn{1}{c}{$\sigma$} & \multicolumn{1}{c}{$N$} & \multicolumn{1}{c}{$T$} & \multicolumn{1}{c}{GAGFL} & \multicolumn{1}{c}{PLS} & \multicolumn{1}{c}{BFK} &       & \multicolumn{1}{c}{GAGFL} & \multicolumn{1}{c}{PLS} & \multicolumn{1}{c}{BFK} \\
			\midrule
			\multicolumn{1}{l}{DGP.1} 
			& 0.5   & 50    & 10    & 0.1161 & 1.4424 & 3.9507 &       & 0.9237 & 0.1091 & 0.5817 \\
			&       & 50    & 20    & 0.0611 & 1.3676 & 4.4311 &       & 0.9349 & 0.0582 & 0.7007 \\
			&       & 50    & 40    & 0.0388 & 1.3680 & 0.7098 &       & 0.9477 & 0.0369 & 0.7126 \\
			&       & 100   & 10    & 0.1022 & 1.4491 & 5.5612 &       & 0.9265 & 0.0531 & 0.5902 \\
			&       & 100   & 20    & 0.0493 & 1.3770 & 5.1966 &       & 0.9408 & 0.0217 & 0.7183 \\
			&       & 100   & 40    & 0.0429 & 1.3803 & 3.0278 &       & 0.9394 & 0.0068 & 0.7146 \\
			&       &       &       &       &       &       &       &       &       &  \\
			& 0.75  & 50    & 10    & 0.2347 & 1.4444 & 8.0057 &       & 0.8396 & 0.1158 & 0.5775 \\
			&       & 50    & 20    & 0.1612 & 1.3693 & 7.3826 &       & 0.8940 & 0.0628 & 0.6807 \\
			&       & 50    & 40    & 0.0771 & 1.3692 & 2.8830 &       & 0.9401 & 0.0398 & 0.7156 \\
			&       & 100   & 10    & 0.1916 & 1.4501 & 3.7916 &       & 0.8860 & 0.0568 & 0.5856 \\
			&       & 100   & 20    & 0.1051 & 1.3778 & 7.5887 &       & 0.9267 & 0.0238 & 0.6948 \\
			&       & 100   & 40    & 0.0467 & 1.3807 & 4.5125 &       & 0.9406 & 0.0086 & 0.7314 \\
			\midrule
			\multicolumn{1}{l}{DGP.2} 
			& 0.5   & 50    & 10    & 0.0787 & 1.4421 & 9.1041 &       & 0.9276 & 0.1050 & 0.5819 \\
			&       & 50    & 20    & 0.0435 & 1.3672 & 1.8993 &       & 0.9441 & 0.0557 & 0.7095 \\
			&       & 50    & 40    & 0.0270 & 1.3676 & 4.5627 &       & 0.9435 & 0.0384 & 0.7005 \\
			&       & 100   & 10    & 0.0708 & 1.4490 & 1.2410 &       & 0.9367 & 0.0498 & 0.5930 \\
			&       & 100   & 20    & 0.0347 & 1.3770 & 2.3579 &       & 0.9444 & 0.0212 & 0.7143 \\
			&       & 100   & 40    & 0.0428 & 1.3802 & 1.6657 &       & 0.9438 & 0.0065 & 0.6984 \\
			&       &       &       &       &       &       &       &       &       &  \\
			& 0.75  & 50    & 10    & 0.1725 & 1.4434 & 28.9436 &       & 0.8826 & 0.1119 & 0.5658 \\
			&       & 50    & 20    & 0.0892 & 1.3683 & 12.6888 &       & 0.9333 & 0.0592 & 0.6948 \\
			&       & 50    & 40    & 0.0718 & 1.3684 & 12.2756 &       & 0.9330 & 0.0395 & 0.7046 \\
			&       & 100   & 10    & 0.1441 & 1.4497 & 10.0073 &       & 0.9152 & 0.0528 & 0.5797 \\
			&       & 100   & 20    & 0.0732 & 1.3774 & 5.8872  &       & 0.9412 & 0.0226 & 0.7115 \\
			&       & 100   & 40    & 0.0321 & 1.3805 & 7.7923  &       & 0.9417 & 0.0071 & 0.7235 \\
			\hline
		\end{tabular}%
	\end{center}
\end{table}

\begin{table}[htp]
	\begin{center}
		Table~\ref{tab:coef} (cont.): Root mean squared error and coverage probability of coefficient estimates\\
		\begin{tabular}{ccccccccccccccc}\hline\hline
			&       &       &       & \multicolumn{3}{c}{RMSE} &       & \multicolumn{3}{c}{Coverage} \\
			\cmidrule{5-7}\cmidrule{9-11}          & \multicolumn{1}{c}{$\sigma$} & \multicolumn{1}{c}{$N$} & \multicolumn{1}{c}{$T$} & \multicolumn{1}{c}{GAGFL} & \multicolumn{1}{c}{PLS} & \multicolumn{1}{c}{BFK} &       & \multicolumn{1}{c}{GAGFL} & \multicolumn{1}{c}{PLS} & \multicolumn{1}{c}{BFK} \\
			\midrule
			\multicolumn{1}{l}{DGP.3} 
			& 0.5   & 50    & 10    & 0.1588 & 1.4535 & 13.3880 &       & 0.8444 & 0.1766 & 0.4305 \\
			&       & 50    & 20    & 0.0747 & 1.3791 & 18.8648 &       & 0.8919 & 0.1280 & 0.5751 \\
			&       & 50    & 40    & 0.0438 & 1.3803 & 1.8619  &       & 0.9175 & 0.1023 & 0.6808 \\
			&       & 100   & 10    & 0.1339 & 1.4539 & 27.3479 &       & 0.8870 & 0.0929 & 0.4392 \\
			&       & 100   & 20    & 0.0573 & 1.3822 & 6.7030  &       & 0.9298 & 0.0539 & 0.5951 \\
			&       & 100   & 40    & 0.0301 & 1.3846 & 15.0779 &       & 0.9319 & 0.0446 & 0.6745 \\
			&       &       &       &        &        &         &       &        &       &  \\
			& 0.75  & 50    & 10    & 0.3212 & 1.4573 & 34.3459 &       & 0.7375 & 0.1877 & 0.4420 \\
			&       & 50    & 20    & 0.2010 & 1.3826 & 25.1017 &       & 0.7744 & 0.1426 & 0.5780 \\
			&       & 50    & 40    & 0.1027 & 1.3839 & 23.9771 &       & 0.8755 & 0.1135 & 0.6787 \\
			&       & 100   & 10    & 0.2428 & 1.4555 & 54.6257 &       & 0.7987 & 0.0990 & 0.4463 \\
			&       & 100   & 20    & 0.1294 & 1.3840 & 12.1645 &       & 0.8806 & 0.0596 & 0.5836 \\
			&       & 100   & 40    & 0.0564 & 1.3861 & 22.4831 &       & 0.9150 & 0.0480 & 0.6837 \\
			\hline
			\multicolumn{1}{l}{DGP.4}
			& 0.5  & 50    & 10     & 0.0478 & 0.7358 & 15.7084 &       & 0.9140 & 0.1170 & 0.5257 \\
			&       & 50    & 20    & 0.0426 & 0.6919 &  0.1803 &       & 0.9416 & 0.0580 & 0.7191 \\
			&       & 50    & 40    & 0.0329 & 0.6909 &  6.1932 &       & 0.9478 & 0.0324 & 0.6311 \\
			&       & 100   & 10    & 0.0317 & 0.7359 & 15.5052 &       & 0.9312 & 0.0474 & 0.5281 \\
			&       & 100   & 20    & 0.0156 & 0.6924 &  0.1819 &       & 0.9386 & 0.0198 & 0.7125 \\
			&       & 100   & 40    & 0.0096 & 0.6936 &  3.5848 &       & 0.9482 & 0.0093 & 0.6429 \\
			&       &       &       &        &        &         &       &       &       &  \\
			& 0.75  & 50    & 10    & 0.0766 & 0.7259 & 10.8228 &       & 0.8771 & 0.1150 & 0.5292    \\
			&       & 50    & 20    & 0.0391 & 0.6827 & 0.2200  &       & 0.9060 & 0.0597 & 0.7391 \\
			&       & 50    & 40    & 0.0368 & 0.6821 & 5.9803  &       & 0.9327 & 0.0320 & 0.4573 \\
			&       & 100   & 10    & 0.0552 & 0.7262 & 23.8383 &       & 0.9107 & 0.0483 & 0.5339 \\
			&       & 100   & 20    & 0.0299 & 0.6833 & 0.2212  &       & 0.9444 & 0.0177 & 0.7345 \\
			&       & 100   & 40    & 0.0178 & 0.6850 & 4.0357  &       & 0.9435 & 0.0091 & 0.6109 \\
			\hline
		\end{tabular}%
	\end{center}
	{\footnotesize\emph{Notes:} For DGP.4, we report only the statistics associated with the estimated coefficient of the lagged dependent variable as the accuracy of the coefficient estimate of the exogenous variable is very similar.}
\end{table}

	\subsection{Extensions of simulation}
	\label{S: sim ext}
	
	We consider four extensions of the simulation designs. This section briefly presents the designs and results of the extensions, while additional details are in Sections S.1.2--1.5 of the supplement.
	
	First, we consider the cases where the regressors are group dependent. Allowing for group-dependent regressors does not affect the clustering accuracy, no matter whether the group structure of regressors coincides with the structure of coefficients, except in DGP.3. In DGP.3 with group-dependent regressors, the misclassification frequency tends to be higher than in the case of independent regressors. 
	
	Second, we consider the case where breaks are small and the groups are more alike. In this case, clustering and break detection become more difficult, and we detect higher misclassification frequency and less accurate estimates of break dates. However, performance rapidly improves as the sample size increases. 
		
	To better understand how cross-sectional variation plays a role in affecting the performance of GAGFL, we consider the case of unequally sized groups, say $N_1:N_2:N_3 = 0.1 : 0.8 : 0.1$, such that some groups contain very few units. 
	In Table~\ref{tab:simulation-extension}, we summarize the average misclassification frequency, the frequency of correct estimation of the number of breaks, and the Hausdorff error of the break date estimates of GAGFL in the presence of small groups.\footnote{To conserve space, Table~\ref{tab:simulation-extension} provides the results in the leading case of DGP.3, while the more full results are in the supplement.} Although small group size results in less accurate estimation due to a lack of cross-sectional variation, we find that increasing the sample size improves the classification accuracy. In particular, accuracy improves significantly as $N$ increases because cross-sectional variation in the small groups is increased, which improves the coefficient estimates and further indirectly improves classification. 
	
	Finally, we consider the case where the break dates are close. We generate breaks in the first group that occur at $\lfloor T/2\rfloor$ and $\lfloor 2T/3\rfloor$, and in the second group at $\lfloor T/3\rfloor $ and $\lfloor T/2\rfloor$, where $\lfloor\cdot\rfloor$ takes the integer part. Now, the difference between the two break dates in both groups is just $\lfloor T/6\rfloor$, i.e. $1$ when $T=10$, $3$ when $T=20$, and $6$ when $T=40$. For the third group, the slope coefficient is stable without a break. The performance of GAGFL is summarized in the bottom panel of Table~\ref{tab:simulation-extension}. This shows that shrinking the interval between the two breaks barely affects the misclassification frequency and the accuracy of break estimation. This indicates that as long as there are sufficient individual units in each group, we can consistently estimate the slope coefficients (and further the groups and breaks), even when the two breaks are consecutive.
	
	We conclude this section by commenting on the iterative feature of the algorithm. On average, the algorithm takes two to three steps to converge in our simulation. The number of iterations increases when $\sigma_\epsilon$ is large but decreases as $T$ increases. It also takes more steps to converge when we generate the data with closer break dates, a small degree of group heterogeneity and breaks, or small groups containing only a few units. Comparing the performance of the iterative and non-iterative estimates, we find the misclassification frequency of the iterative estimates consistently lower than the rate of the non-iterative estimates, suggesting that iteration improves clustering performance, sometimes greatly. 
	Consequently, iterative estimates produce lower RMSEs and higher coverage probabilities than non-iterative estimates.

	\begin{table}[t]
		\begin{center}
			\caption{Simulation extensions: GAGFL clustering and break detection (DGP.3)}\label{tab:simulation-extension}
			\begin{tabular}{llllcccccccccc}
				\hline\hline
				& \multicolumn{3}{c}{$N=50$} && \multicolumn{3}{c}{$N=100$}\\
				& $T=10$ & $T=20$  & $T=40$ && $T=10$ & $T=20$ & $T=40$\\
				\hline
				
				&\multicolumn{7}{c}{Groups with few units}\\
				Misclassification frequency           & 0.1872   & 0.1816   & 0.1133 && 0.1308  & 0.1283  & 0.1145\\  \\
				\multicolumn{7}{l}{Freq. of correct estimation of $m_{g}$}  \\
				$G_1$ ($m_{1,0}^0=2$)                & 0.2840   & 0.3120   & 0.3480 && 0.5320  & 0.6140  & 0.7020 \\
				$G_2$  ($m_{2,0}^0=2$)               & 0.8820   & 0.9500   & 0.9760 && 0.9220  & 0.9920  & 0.9920 \\
				$G_3$  ($m_{3,0}^0=0$)               & 0.0880   & 0.2320   & 0.5640 && 0.3660  & 0.5080  & 0.6820 \\ \\
				\multicolumn{7}{l}{Hausdorff error}  \\                               
				$G_1$ ($m_{1,0}^0=2$)                & 0.1439  & 0.1407  & 0.1267&& 0.0826 & 0.0581 & 0.0465\\
				$G_2$ ($m_{2,0}^0=2$)                & 0.0088  & 0.0029  & 0.0000&& 0.0060 & 0.0001 & 0.0005\\

				\hline
				
				&\multicolumn{7}{c}{Closer break dates} \\
				Misclassification frequency           & 0.0119  & 0.0004  & 0.0000&& 0.0095 & 0.0005 & 0.0000\\\\
				\multicolumn{7}{l}{Freq. of correct estimation of $m_{g}$}  \\
				$G_1$ ($m_{1,0}^0=2$)                & 0.8080   & 0.9120   & 0.9620 && 0.9520  & 0.9820  & 1.0000 \\
				$G_2$  ($m_{2,0}^0=2$)               & 0.8180   & 0.9040   & 0.9760 && 0.9620  & 0.9940  & 0.9980 \\
				$G_3$  ($m_{3,0}^0=0$)               & 0.7560   & 0.9280   & 0.9960 && 0.9240  & 1.0000  & 1.0000 \\ \\
				\multicolumn{7}{l}{Hausdorff error}  \\                               
				$G_1$ ($m_{1,0}^0=2$)                & 0.0335  & 0.0136  & 0.0036&& 0.0082 & 0.0037 & 0.0000\\
				$G_2$ ($m_{2,0}^0=2$)                & 0.0405  & 0.0158  & 0.0031&& 0.0084 & 0.0015 & 0.0001\\ 			
				\hline
				
				\hline
			\end{tabular}
		\end{center}
	\end{table}

	\section{Empirical application}\label{sec:empirical}
	We apply the proposed GAGFL estimator to revisit the relationship between democracy and income. This analysis was first undertaken by \citet{acemoglu&johnson&robinson&yared2008} in a standard panel framework, and revisited by \citet[BM hereafter]{bonhomme&manresa2015} using the GFE approach. As suggested by \citet{acemoglu&johnson&robinson&yared2008} and BM, countries experience economic and democratic development at certain ``critical junctures'', such as the end of feudalism, the industrialization age, or the process of colonization. These junctures may not only influence the average degree of democracy (captured by the intercept of the democracy--income regression), but also the degree of democracy persistence and the relationship between income and democracy (captured by the slope coefficients). This implies that the effect of income on democracy can shift discontinuously at those junctures. Moreover, the paths of economic and political development also diverge across countries. For example, industrialization may affect a proportion of Western countries, but not Asian countries, at least to a lesser extent or at a later stage. Hence, the democracy--income relationship is likely to shift at different historical junctures across countries, and even for countries affected by the same event, the extent of the effect can be distinct.
	
	We revisit the democracy--income relation by allowing heterogeneous structural breaks in a regression of democracy, measured by the Freedom House index, on its lagged value and the lagged value of the logarithm of GDP per capita, namely
	\begin{equation}
	\label{eq:demo-income-model}
	\textrm{democracy}_{it} = \alpha_{g_it}+\theta_{1,g_i, t}\textrm{democracy}_{i,t-1}+\theta_{2,g_i,t}\textrm{income}_{i,t-1}+\varepsilon_{it}.
	\end{equation}
One essential difference from BM's specification is that we allow the slope coefficients to have a grouped pattern and possible structural breaks. Note also that in our specification, $\alpha_{g_i t } $ may exhibit structural breaks and may not change at every time period. We compare our results with those of BM.\footnote{\citet{lu&su2017} also studied the same empirical data with group-specific slope coefficients, but they considered individual and time-specific two-way fixed effect panel models.}

	We follow BM in using a balanced subsample of the data of \citet{acemoglu&johnson&robinson&yared2008} that contains 90 countries for seven periods (five-year frequency over 1970--2000).
	To ensure similar variation across variables for grouping, we standardize democracy and income by subtracting their overall means and then divide them by the overall standard deviations, respectively.
	
	To implement GAGFL, we choose $\lambda_{\max} = 50$, which results in zero breaks in all groups, and $\lambda_{\min}=0.001$, which results in six breaks in all groups. We then choose the tuning parameter $\lambda$ by searching on the interval $[\lambda_{\min},\lambda_{\max}]$ with 200 evenly-distributed logarithmic grids. We follow \citet{qian&su2016} and use the same information criterion as in the simulation with $\rho_{N T} = 0.05\ln(NT)/\sqrt{NT}$ for determining the number of breaks. 
	To determine the number of groups, we employ the BIC as in the simulation. We let the number of groups vary from 1 to 10, and the minimum information criterion corresponds to four groups (i.e. $G=4$), which coincides with the group number specification of BM. 

	When we estimate~\eqref{eq:demo-income-model} with four groups, we have three groups with structural changes in intercepts and slope coefficients and one group without any break, and name them Groups 4.1--4.4. 
	Figure~\ref{fig:group-pattern-G4} displays the estimated grouped pattern when $G=4$, and the estimates of coefficients and structural breaks are reported in Table~\ref{tab:coefficient-G4}. We provide the post-Lasso estimates with their corresponding standard errors. The supplement presents the confidence sets of group membership using the method in \citet{DzemskiOkui2018}. The unit-wise confidence sets suggest that the group membership estimates are reasonably accurate, although the joint sets are wide.

	\begin{figure}[t]\caption{Estimates of group membership ($G=4$)}\label{fig:group-pattern-G4}\ \\
		\centering
		\includegraphics[width=\linewidth]{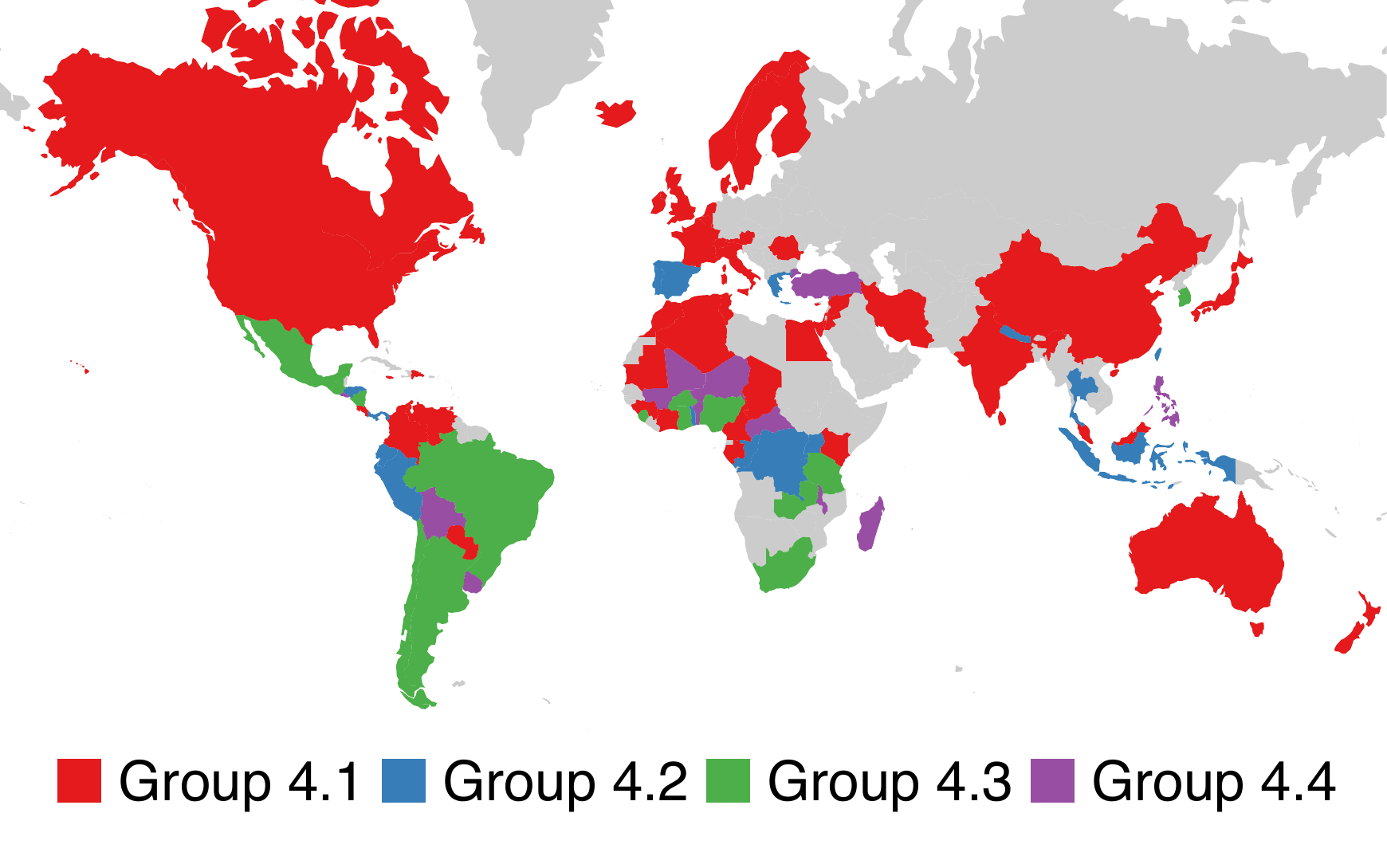}
	\end{figure}
\begin{table}[t]
	\centering
	\caption{Income and democracy: Coefficient and regime estimates of $G=4$}\label{tab:coefficient-G4}\small
	\begin{tabular}{rrrrrrrrr}
		\hline\hline
		& \multicolumn{1}{l|}{Regime} & \multicolumn{1}{c|}{1} & \multicolumn{1}{c|}{2} & \multicolumn{1}{c|}{3} & \multicolumn{1}{c|}{4} & \multicolumn{1}{c|}{5} & \multicolumn{1}{c|}{6} & \multicolumn{1}{c}{7} \\
		\hline
		\multicolumn{1}{l}{Group 4.1} & \multicolumn{1}{l|}{Intercept} & \multicolumn{7}{c}{$-0.0644$ $(0.0229)$} \\
		& \multicolumn{1}{l|}{Democracy$_{t-1}$} & \multicolumn{7}{c}{$0.8655$ $(0.0268)$} \\
		& \multicolumn{1}{l|}{Income$_{t-1}$} & \multicolumn{7}{c}{$0.1404$ $(0.0320)$} \\
		\hline
		\multicolumn{1}{l}{Group 4.2} & \multicolumn{1}{l|}{Intercept} & \multicolumn{1}{c|}{$-1.1291$ $( 0.0738)$} & \multicolumn{6}{c}{$0.1713$ $( 0.0472)$} \\
		& \multicolumn{1}{l|}{Democracy$_{t-1}$} & \multicolumn{1}{c|}{$-0.3189$ $(0.0935)$} & \multicolumn{6}{c}{$0.5172$ $( 0.0773)$} \\
		& \multicolumn{1}{l|}{Income$_{t-1}$} & \multicolumn{1}{c|}{$0.2590$ $(0.0645)$} & \multicolumn{6}{c}{$0.4172$ $(0.0756)$} \\
		\hline
		\multicolumn{1}{l}{Group 4.3} & \multicolumn{1}{l|}{Intercept} & \multicolumn{4}{c|}{$-0.3807$ $(0.1070)$} & \multicolumn{3}{c}{$0.2686$ $(0.0750)$} \\
		& \multicolumn{1}{l|}{Democracy$_{t-1}$} & \multicolumn{4}{c|}{$-0.1574$ $(0.1694)$} & \multicolumn{3}{c}{$0.2477$ $(0.1520)$} \\
		& \multicolumn{1}{l|}{Income$_{t-1}$} & \multicolumn{4}{c|}{$0.3425$ $(0.0925)$} & \multicolumn{3}{c}{$0.4239$ $(0.1450)$} \\
		\hline
		\multicolumn{1}{l}{Group 4.4} & \multicolumn{1}{l|}{Intercept} & \multicolumn{3}{c|}{$-0.4265$  $(0.1370)$} & \multicolumn{2}{c|}{$0.5290$ $(0.0914 )$} & \multicolumn{2}{c}{$0.4983$ $(0.1690)$} \\
		& \multicolumn{1}{l|}{Democracy$_{t-1}$} & \multicolumn{3}{c|}{$0.4796$ $(0.0884)$} & \multicolumn{2}{c|}{$0.3879$ $(0.1708)$} & \multicolumn{2}{c}{$0.0860$ $(0.0949)$} \\
		& \multicolumn{1}{l|}{Income$_{t-1}$} & \multicolumn{3}{c|}{$0.3019$ $(0.1235)$} & \multicolumn{2}{c|}{$0.8562$ $(0.1773)$} & \multicolumn{2}{c}{$-0.0551$ $(0.1495)$} \\
		\hline
		&       &       &       &       &       &       &       &  \\
	\end{tabular}%
\end{table}%
	Group 4.1 contains no breaks. 
	A significant feature of this group is strong dynamic persistence in the political system, while the income effect is relatively weak. This group contains a large number of the ``high-democracy'' group countries in BM such as the US and Switzerland. Nevertheless, the group also includes many countries in BM's ``low-democracy'' group, such as China, Iran, Cameroon, Guinea and several other African countries. This may appear counterintuitive at first glance, but there are both developed and developing countries in this group whose unit-wise confidence sets (see the supplement) are singleton, suggesting that this result is not an artifact of statistical error. Further examination reveals that this group structure is mainly driven by the persistence of democracy. Although democracy levels within Group 4.1 vary, a common feature is that their political systems are highly persistent, as reflected in the high value of $\theta_{1}$ (= 0.8655). This strong persistency separates countries of this group from other groups containing breaks in their democracy level. When we set $G=6$ corresponding to the second-lowest value of BIC, this group will be segmented (see the supplement). A large discrepancy in democracy and income levels across countries also explains the small intercept and weak income effect, $\theta_{2}$. 
	
	Group 4.2 is characterized by one structural break occurring at an early stage of the sample period (the mid-late 1970s). The average democracy level, persistence, and income effect all increase after this break. This group consists mainly of ``early transition'' countries in BM such as Greece, Nepal, Spain, and Thailand. It also contains a few ``low-democracy'' countries, including Burundi, Republic of Congo, Togo, etc. A further examination of these low-democracy countries shows that they all have a break in their democracy level at the beginning of the sample period, even though the level remains low in general. In this sense, the break and coefficient estimates produced by GAGFL well capture the political junctures in the mid-late 1970s of these countries.

	Group 4.3 also exhibits one structural break, but in the mid- to late 1980s. Again, the average democracy level, persistence, and income effect increase after the break, but to a lesser extent when compared with Group 4.2. 
	A large proportion of this group are ``late transition'' countries whose democratic reforms occurred at a later stage. It also includes countries that transitioned to highly democratic countries at multiple junctures (e.g., South Korea, Argentina, and Brazil) and those whose democracy level fluctuated greatly in the latter period (i.e., Sierra Leone and Nigeria). 
	
	Finally, Group 4.4 exhibits two breaks, one at the beginning of the 1980s and one at the beginning of the 1990s. After each break, dynamic democratic persistence decreases from 0.4796 to 0.3879 and then to 0.0860, and the income effect first strengthens from 0.3019 to 0.8562, and then weakens with an insignificant estimate $-0.0551$. Most countries in this group experienced changes/fluctuations in democracy in the middle and late phase of the period. However, none has a singleton confidence set (see the supplement) and thus it might be difficult to make a strong argument about this group. 
	
	We also examine the democracy--income relation under alternative specifications. First, we consider the results with $G=6$ which is the second-lowest value of BIC.
	In this case, the group with stable coefficients under $G=4$, i.e. Group 4.1, is further divided. Recall that Group 4.1 contains countries with different levels of democracy. When we set $G=6$, the lowly and highly democratic countries are clearly separated (e.g., China and the US), although their slope coefficients are all stable over time. 
	Second, we consider the setup with a fully time-varying intercept, where only slope coefficients are penalized but the $\alpha_{g_i,t}$s are not. The results are generally in line with those of regime-specific estimation (in Table~\ref{tab:coefficient-G4}). These indicate that average democracy does not fluctuate in every time period, but exhibits at most one discontinuous break during the sample period. 
	Moreover, the fully time-varying intercepts are much less significant than those reported in Table~\ref{tab:coefficient-G4}.
		Thus, in this particular application, the fully time-varying estimation may be rather inefficient compared to regime-specific estimation. We provide the complete results of these alternative specifications in the supplement.

	In general, our procedure detects a large degree of heterogeneity in the income effect of democracy.
	Although compatible to some extent, our grouping differs from \citet{bonhomme&manresa2015}. This is mainly because we form groups based on the whole coefficient vector and allow structural change. The grouped pattern thus not only reflects the magnitude difference of the coefficient estimates, but also the heterogeneity in structural breaks. This allows us to provide additional insights not captured by GFE. In particular, we show that the persistence of the democracy level and the income effect both exhibit significant structural breaks in several of the groups, and that the breakpoints and the size of breaks differ markedly across groups.
	
	We also apply GAGFL to another application to study the determinants of the cross-country differences in savings behavior as in \citet{su&shi&phillips2016} and the analysis is in the supplement. We find more heterogeneity than \citet{su&shi&phillips2016} in the sense that countries differ in their transition points, the level and stability of their savings rates, and the effects of various determinants. One group of Asian countries features one structural break coinciding with the start of the 1997 Asian financial crisis, while the crisis hardly affects the remaining countries, and they thus have relatively stable coefficients. This application again confirms the importance of incorporating heterogeneous structural breaks. 
	
	\section{Conclusion}\label{sec:conclusion}
	
	Structural change in the relationships between variables often characterizes large panels with long time series because of important events, such as financial crises, technological progress, economic transition, etc. The effect of these events often differs across individuals. Some can be greatly influenced, while others in the sample may not be affected at all by events. Even for those affected, the impacts can be heterogeneous. 
	Importantly, failing to incorporate heterogeneity in structural breaks leads to incorrect breakpoint detection and imprecise coefficient estimates.
	
	In this paper, we propose a new model and estimation procedure for panel data with heterogeneous structural breaks. We model individual heterogeneity via a \emph{latent} grouped pattern such that individuals within the group share the same regression coefficients. For each group, we allow common structural breaks in the coefficients, while the number of breaks, the breakpoints, and/or break sizes can differ across groups. We develop a hybrid procedure of the GFE estimator and AGFL to estimate the model. With the proposed procedure, we can 1) consistently estimate the latent group membership structure, 2) automatically determine the number of breaks and consistently estimate the breakpoints for each group, and 3) consistently estimate the regression coefficients with group-specific structural breaks.
	
	An interesting extension of our work would be to allow the grouped pattern to change at disjoint time intervals. Such a flexible framework is desirable because the impact of significant events may completely change the economic mechanisms governing individuals and their group memberships. For example, the impact of the most recent global financial crisis was so enormous that it reshaped economies and thus changed the group structure of the world. One of the main difficulties is to restrict appropriately the structural breaks, such that there are sufficient observations within each regime to estimate the grouped pattern. This presents a challenge for future research.

	\appendix
	
	\section{Technical appendix}
	
	This appendix provides the proofs of the technical results in the main text.
	First, we discuss the asymptotic properties of the GFE-type estimator in Section \ref{sec-gfe-proof}.
	Note that this estimator is the preliminary estimator for our GAGFL procedure.
	We then establish the asymptotic properties of the AGFL estimator applied to each group
	in Section \ref{sec-agfl-proof}.
	The asymptotic properties of GAGFL are discussed in \ref{sec-agafl-proof}.
	They are based on the results in Sections \ref{sec-gfe-proof} and \ref{sec-agfl-proof}
	
	\subsection{Asymptotic properties of the GFE-type estimator}
	\label{sec-gfe-proof}
	
	We extend the results in \citet{bonhomme&manresa2015} to models with group-specific and time-varying coefficients. The arguments given in this section are very similar to those in \citet{bonhomme&manresa2015}. We thus present only the results and defer the proofs of the lemmas to the supplement.

	We first introduce some useful lemmas. Denote $M$ as a generic universal constant. Let 
	$$
	\dot Q_{NT} (\beta, \gamma)
	= \frac{1}{NT}
	\sum_{i=1}^N \sum_{t=1}^T
	(y_{it}  - x_{it} ' \beta_{g_i ,t} )^2,
	$$
	and
	\begin{align*}
	\tilde Q_{NT} (\beta, \gamma)
	= \frac{1}{NT}
	\sum_{i=1}^N \sum_{t=1}^T
	( x_{it}' (\beta_{g_i^0, t}^0
	-  \beta_{g_i, t}))^2
	+ \frac{1}{NT}
	\sum_{i=1}^N \sum_{t=1}^T
	\epsilon_{it}^2.
	\end{align*}

	\begin{lemma}
		\label{lem-qd-uniform}
		Suppose that Assumptions \ref{a: basic}.\ref{as-compact}--\ref{as-ijcov-bound} hold. Then, we have
		\begin{align*}
		\sup_{(\beta, \gamma) \in \mathcal{B}^{GT} \times \Gamma_{G}}
		\left| \dot Q_{NT} (\beta, \gamma) - \tilde Q_{NT}
		(\beta, \gamma)
		\right|
		= o_p(1).
		\end{align*}
	\end{lemma}

	We consider the following HD in $\mathcal{B}^{GT}$ such that
	\begin{align*}
	d_H ( \beta^a , \beta^b)
	= \max \left\{ \max_{g \in \mathbb{G}}
	\left( \min_{\tilde g \in \mathbb{G}}
	\frac{1}{T} \sum_{t=1}^T \left\Vert \beta_{\tilde g, t}^a - \beta_{g, t}^b \right\Vert^2
	\right),
	\max_{ \tilde g \in \mathbb{G}}
	\left( \min_{ g \in \mathbb{G}}
	\frac{1}{T} \sum_{t=1}^T \left\Vert \beta_{\tilde g, t}^a - \beta_{g, t}^b \right\Vert^2
	\right)
	\right\}.
	\end{align*}

	\begin{lemma}
		\label{lem-hd-consistent}
		Suppose that Assumptions \ref{a: basic}.\ref{as-compact}--\ref{as-ijcov-bound} and \ref{a: iden} hold. Then, we have $
		d_H ( \dot \beta , \beta^0) = o_p(1),
		$
		where $\dot \beta$ is defined in \eqref{eq-gfe}.
	\end{lemma}

	The proof of Lemma \ref{lem-hd-consistent} shows that there exists a permutation $\sigma$ such that \\ $ \sum_{t=1}^T \left\Vert \beta_{ \sigma (g), t }^0 - \dot \beta_{g, t} \right\Vert^2 /T
	= o_p(1)$.
	We obtain $\sigma (g) =g$ by relabeling.

	Define
	$
	\mathcal{N}_{\eta}
	= \left\{ \beta \in \mathcal{B}^{GT}
	: 1/T \sum_{t=1}^T \left\Vert \beta_{g,t}^0 - \beta_{g,t} \right\Vert^2 < \eta,
	\forall g \in \mathbb{G}
	\right\}.
	$
	Let
	\begin{align}
	\hat g_i (\beta)
	= \arg \min_{g \in \mathbb{G}} \sum_{t=1}^T
	(y_{it}  - x_{it} ' \beta_{g ,t} )^2.
	\label{eq-hatg-def}
	\end{align}

	\begin{lemma}
		\label{lem-g-consistent}
		Suppose that Assumptions \ref{a: iden}.\ref{as-d-bound} and \ref{a: tail} are satisfied.
		For $\eta > 0$ small enough, we have $\forall \delta > 0$,
		$
		\sup_{\beta \in \mathcal{N}_{\eta}}
		1/N \sum_{i=1}^N
		\mathbf{1}\left\{ \hat g_i (\beta) \neq g_i^0 \right\}
		= o_p ( T^{-\delta}).
		$
	\end{lemma}

	Let	
	$
	\check \beta
	= \arg\min_{\beta \in \mathcal{B}^{GT}}
	\sum_{i=1}^N \sum_{t=1}^T
	(y_{it}  - x_{it} ' \beta_{g_i^0 ,t} )^2.
	$
	Note that $\check \beta$ is the estimator of $\beta$ when the group memberships (i.e., $\gamma^0$) are known.

	\begin{lemma}
		\label{lem-cb}
		Suppose that Assumptions \ref{a: basic} and \ref{a: iden}.\ref{as-m-eigen} hold. Suppose that $N_g /N \to \pi_g > 0$ for any $g \in \mathbb{G}$.
		Then, it follows that for all $g$ and $t$,
		$
			\check \beta_{g,t} - \beta_{g,t}^0 = O_p\left(1/\sqrt{N}\right).
		$
	\end{lemma}

	\begin{lemma}
		\label{lem-db-cb}
		Suppose that Assumptions \ref{a: basic}, \ref{a: iden},
		and \ref{a: tail} are satisfied. As $N,T \to \infty$, for any $\delta >0$,
		it holds that
		$
		\dot \beta_{g,t} = \check \beta_{g,t} + o_p (T^{-\delta}),
		$
		for all $g$ and $t$.
	\end{lemma}

	We can now consider the rate of convergence of the elements of $\dot \beta$.

	\begin{proof}[Proof of Theorem \ref{thm-gfe-ad}]
	
		The theorem follows from Lemma \ref{lem-cb} and Lemma \ref{lem-db-cb}.
	\end{proof}

	\subsection{Asymptotic properties of the GAGFL estimator with known group membership}
	\label{sec-agfl-proof}

	The arguments here are very similar to those in \citet{qian&su2016}. While their results do not cover our case, our setting is indeed simpler because we do not have individual-specific intercepts. We thus present only the results and leave the proofs to the supplement.
	
	Let
	\begin{align*}
	\mathring Q (\beta)
	= \frac{1}{NT}\sum_{i=1}^N \sum_{t=1}^T
	(y_{it}  - x_{it} ' \beta_{g_i^0 ,t} )^2
	+
	\lambda  \sum_{g\in \mathbb{G}}
	\sum_{t=2}^T \dot w_{g,t} \left\Vert \beta_{g, t} - \beta_{g, t-1}
	\right\Vert .
	\end{align*}
	Note that $\mathring Q(\beta) = \hat Q (\beta, \gamma^0)$
	and that $\mathring \beta = \arg \min_{\beta \in \mathcal{B}^{GT}} \mathring Q (\beta)$.
	We derive the asymptotic distribution of $\mathring \beta$.

	\begin{lemma}
		\label{lem-lasso-c}
		Suppose that Assumptions \ref{a: basic}.\ref{as-ij-bound}, \ref{a: iden}.\ref{as-m-eigen}, and \ref{a: break}.\ref{as-jmin} hold. Suppose that $N_g /N \to \pi_g > 0$ for any $g \in \mathbb{G}$.
		We have, as $N,T \to \infty$,
		$
		\frac{1}{T} \left\| \mathring \beta_g - \beta_g^0 \right\|^2= O_p \left(1/N\right)
		$
		for any $g\in \mathbb{G}$.
		We also have, as $N,T \to \infty$,
		$
		\mathring \beta_{g,t} - \beta_{g,t}^0 = O_p \left(1/\sqrt{N}\right).
		$
		
	\end{lemma}
	
	Next, we consider the difference between coefficient estimates in two consecutive periods. Let $\mathring \theta_{g,1} = \mathring \beta_{g,1} $ and $\mathring \theta_{g,t} = \mathring \beta_{g,t} - \mathring \beta_{g, t-1}$.
	
	\begin{lemma}
		\label{lem-break-c}
		Suppose that Assumptions \ref{a: basic}.\ref{as-ij-bound}, \ref{a: iden}.\ref{as-m-eigen}, \ref{a: basic}.\ref{as-x-4moment}, and \ref{a: break} hold. Suppose that $N_g /N \to \pi_g > 0$ for any $g \in \mathbb{G}$.
		It follows that
		$
		\Pr \left( \left\| \mathring \theta_{g,t} \right\| = 0, \forall t \in \mathcal{T}_{m_g^0,g}^{0c}, g \in \mathbb{G}\right) \to 1
		$
		as $N \to \infty$.
	\end{lemma}

	\begin{lemma}
		Suppose that Assumptions \ref{a: basic}.\ref{as-ij-bound}, \ref{a: iden}.\ref{as-m-eigen}, \ref{a: basic}.\ref{as-x-4moment}, and \ref{a: break} hold. Suppose that $N_g /N \to \pi_g > 0$ for any $g \in \mathbb{G}$.
		It holds that, as $N\to \infty$,
		$
		\Pr (\mathring m_g = m_{g}^0, \forall g \in \mathbb{G}) \to 1,
		$
		and
		\begin{align*}
		\Pr (\mathring T_{g,j} = T_{g,j}^0, \forall j \in \{1, \dots, m^0\}, g \in \mathbb{G} \mid \mathring m_g =m_g^0, \forall g \in \mathbb{G}) \to 1.
		\end{align*}
	\end{lemma}

	We now obtain the asymptotic distribution of $\mathring \beta$.
	Let $\mathring \alpha_{g, j} = \mathring \beta_{g, t}$ for $T_{g, j}^0 \le t < T_{g, j+1}^0$.

	\begin{lemma}
		\label{lem-mrb-ad}
		Suppose that Assumptions \ref{a: basic}.\ref{as-ij-bound}, \ref{a: iden}.\ref{as-m-eigen}, \ref{a: basic}.\ref{as-x-4moment}, \ref{a: break}, \ref{as-cb-break}, and \ref{as-jmin-imin} hold. Suppose that $N_g /N \to \pi_g > 0$ for any $g \in \mathbb{G}$.
		Let $A$ be a diagonal matrix whose diagonal elements are 
		$(I_{1,1}, \dots, I_{1,m_1^0+1}, I_{2,1}, \dots, I_{2,m_2^0 +1}, I_{3, 1} \dots, I_{G-1, m_{G-1}^0+1}, I_{G, 1}, \dots, I_{G,m_G^0+1})$.
		Let $\Pi$ be a $\sum_{g=1}^G (m_g^0 +1)  k \times \sum_{g=1}^G (m_g^0 +1) k $
		block diagonal matrix whose $g$-th diagonal block is a $(m_g^0+1)k \times (m_g^0 +1)k$ diagonal matrix whose diagonal elements are $\pi_g$.
		
		Conditional on $\mathring m_g = m_g^0$ for all $g \in \mathbb{G}$,
		we have, if $(\max_{g\in \mathbb{G}} m_g^0 )^2 / (I_{\min} \min_{g\in \mathbb{G}} N_g) \to 0 $,
		\begin{align*}
		D \sqrt{N} A^{1/2}(\mathring \alpha - \mathring \alpha^0)
		\to_d N( 0, D\Sigma_{x}^{-1} \Pi^{-1/2} \Omega \Pi^{-1/2} \Sigma_x^{-1} D').
		\end{align*}
	\end{lemma}

	\subsection{Asymptotic properties of the GAGFL estimator with unknown group membership}
	\label{sec-agafl-proof}
	
	Based on the results in Appendixes \ref{sec-gfe-proof} and \ref{sec-agfl-proof}, we now present the asymptotic properties of the GAGFL estimator under \emph{unknown} group membership, which are the main results of the paper. We abbreviate the Cauchy--Schwarz inequality as the CS inequality.

	\subsubsection{Extra lemmas needed} Before we provide the proofs of the lemma and theorems in the paper, we first present some extra lemmas needed. 
	Let
	\begin{align*}
	\hat Q_{N T} (\beta, \gamma)
	= \frac{1}{N T}
	\sum_{i=1}^N \sum_{t=1}^T
	(y_{it} - x_{it} '\beta_{g_i ,t} )^2
	+ \lambda  \sum_{g\in \mathbb{G}}
	\sum_{t=2}^T
	\dot w_{g,t} \left\Vert \beta_{g, t} - \beta_{g, t-1}
	\right\Vert.
	\end{align*}
	and	
	\begin{align*}
	\tilde{\hat Q}_{N T} (\beta, \gamma)
	= \frac{1}{N T}
	\sum_{i=1}^N \sum_{t=1}^T
	(x_{it} ' (\beta_{g_i^0 ,t}^0 - \beta_{g_i,t}))^2
	+ \lambda  \sum_{g\in \mathbb{G}}
	\sum_{t=2}^T
	\dot w_{g,t} \left\Vert \beta_{g, t} - \beta_{g, t-1}
	\right\Vert.
	\end{align*}
	
	\begin{lemma}
		\label{lem-qh-uniform}
		Suppose that Assumption \ref{a: basic} hold. Then, we have
		\begin{align*}
		\sup_{(\beta, \gamma) \in \mathcal{B}^{GT} \times \Gamma_{G}}
		\left| \hat Q_{NT} (\beta, \gamma) - \tilde{\hat Q}_{NT}
		(\beta, \gamma)
		\right|
		= o_p(1).
		\end{align*}
	\end{lemma}
	
	\begin{proof}
		Note that
		$
		\hat Q_{NT} ( \beta, \gamma)
		- \tilde{\hat Q}_{NT} (\beta, \gamma)
		= \dot Q_{NT} (\beta, \gamma )
		- \tilde Q_{NT} (\beta, \gamma ).
		$
		Lemma \ref{lem-qd-uniform} implies the desired result.
	\end{proof}
	
	\begin{lemma}
		\label{lem-hd-consistent-h}
		Suppose that Assumptions \ref{a: basic}, \ref{a: iden}, and \ref{a: break}.\ref{as-jmin}
		hold. Suppose that $N_g /N \to \pi_g > 0$ for any $g \in \mathbb{G}$. It holds that as $N,T \to \infty$,
		$
		d_H ( \hat \beta , \beta^0) = o_p(1).
		$
	\end{lemma}
	
	\begin{proof}
		From Lemma \ref{lem-qh-uniform}, we have
		\begin{align*}
		\tilde{\hat Q} (\hat \beta, \hat \gamma  )
		= \hat Q (\hat \beta, \hat \gamma  ) + o_p(1)
		\le \hat Q ( \beta^0,  \gamma^0  ) + o_p(1)
		= \tilde{\hat Q}( \beta^0,  \gamma^0) + o_p (1).
		\end{align*}
		Because $\tilde Q (\beta, \gamma)$ is minimized at $\beta = \beta^0$ and $\gamma = \gamma^0$,
		we have
		$
		\tilde{\hat Q}  (\hat \beta, \hat \gamma  )
		- \tilde{\hat Q} ( \beta^0,  \gamma^0) = o_p (1).
		$	
		On the other hand, we have
		\begin{align*}
		&\tilde{\hat Q} (\beta, \gamma )  -
		\tilde{\hat Q} (\beta^0, \gamma^0 )\nonumber \\
		=&
		\frac{1}{NT} \sum_{i=1}^N \sum_{t=1}^T \left( x_{it}'( \beta_{g_i^0, t}^0 - \beta_{g_i, t})\right)^2 + \lambda  \sum_{g\in \mathbb{G}}\sum_{t=2}^T \dot w_{g,t}
		\left( \left\Vert \beta_{g, t} - \beta_{g, t-1}
		\right\Vert
		-
		\left\Vert \beta_{g, t}^0 - \beta_{g, t-1}^0
		\right\Vert \right) \nonumber\\
		=&
		\sum_{g=1}^G \sum_{\tilde g=1}^G \frac{1}{T}
		\left( \beta_{g}^0 - \beta_{\tilde g}\right)' M( \gamma, g, \tilde g)
		\left( \beta_{g}^0 - \beta_{\tilde g}\right) + \lambda  \sum_{g\in \mathbb{G}}\sum_{t \in \mathcal{T}_{m_g^0 , g}^0} \dot w_{g,t}
		\left( \left\Vert \beta_{g, t} - \beta_{g, t-1}
		\right\Vert
		-
		\left\Vert \beta_{g, t}^0 - \beta_{g, t-1}^0
		\right\Vert \right) \nonumber \\
		&+ \lambda  \sum_{g\in \mathbb{G}} \sum_{t \in \mathcal{T}_{m_g^0, g}^{0c}} \dot w_{g,t} \left\Vert
		\beta_{g, t} - \beta_{g, t-1}
		\right\Vert.
		\end{align*}
		
		We now examine the three terms on the left-hand side of the previous display in turn.
		For the first term, in the proof of Lemma \ref{lem-hd-consistent}, we have shown that
		\begin{align*}
		\sum_{g=1}^G \sum_{\tilde g=1}^G \frac{1}{T}
		\left( \beta_{g}^0 - \beta_{\tilde g}\right)' M( \gamma, g, \tilde g)
		\left( \beta_{g}^0 - \beta_{\tilde g}\right)
		\ge
		\hat \rho
		\max_{ g \in \mathbb{G}}
		\left(  \min_{\tilde g \in \mathbb{G}}\left( \frac{1}{T} \sum_{t=1}^T \left\Vert \beta_{g, t }^0 - \beta_{\tilde g, t} \right\Vert^2 \right)
		\right).
		\end{align*}
	For the second term, we have, by the Jensen, triangular and CS inequalities,
		\begin{align*}
		&  \left| \lambda  \sum_{g\in \mathbb{G}}\sum_{t \in \mathcal{T}_{m_g^0 , g}^0} \dot w_{g,t}
		\left( \left\Vert \beta_{g, t} - \beta_{g, t-1}
		\right\Vert
		-
		\left\Vert \beta_{g, t}^0 - \beta_{g, t-1}^0
		\right\Vert \right) \right| \\
		\le &
		\lambda  \sum_{g\in \mathbb{G}} \sum_{t \in \mathcal{T}_{m_g^0, g}^0} \dot w_{g,t}  \left\Vert
		\beta_{g, t} - \beta_{g, t-1}
		- (\beta_{g, t}^0 - \beta_{g, t-1}^0)
		\right\Vert \\
		\le &
		\lambda  \max_{s\in \mathcal{T}_{m_g^0, g}^0, g \in \mathbb{G} } ( \dot w_{g,s})
		\sum_{t \in \mathcal{T}_{m_g^0, g}^0}
		\left\Vert \beta_{g, t} - \beta_{g, t-1}
		- (\beta_{g, t}^0 - \beta_{g, t-1}^0)
		\right\Vert  \\
		=& O_p \left( \lambda \left( \sum_{g\in \mathbb{G}} m_g^0 \right) J_{\min}^{-\kappa}\right)
		= o_p(1),
		\end{align*}
		where the last equality follows from Assumptions \ref{a: basic}.\ref{as-compact} and \ref{a: break}.\ref{as-jmin}.	
		Finally, for the third term, note that
		$
		\lambda  \sum_{t \in \mathcal{T}_{m_g^0, g}^{0c}} \dot w_{g,t} \left\Vert
		\beta_{g, t} - \beta_{g, t-1}
		\right\Vert \ge 0
		$. 
	
		Therefore, it holds that 
		\begin{align*}
		o_p (1) = \tilde{\hat Q} (\hat \beta, \hat \gamma )  -
		\tilde{\hat Q} (\beta^0, \gamma^0 ) \ge \rho 	\max_{ g \in \mathbb{G}} \left( \min_{\tilde g \in \mathbb{G}}\left( \frac{1}{T} \sum_{t=1}^T \left\Vert \beta_{g, t }^0 - \hat \beta_{\tilde g, t} \right\Vert^2 \right) \right) + o_p (1).
		\end{align*}
		Because $\hat \rho$ is asymptotically bounded away from zero by Assumption \ref{a: iden}.\ref{as-m-eigen}, we have the desired result.
	\end{proof}
	
	As in the case for $\dot \beta$, the above result implies that
	there exists a permutation $\sigma$ such that
	$
	1/T\sum_{t=1}^T \left\Vert \beta_{ \sigma (g), t }^0 - \hat  \beta_{g, t} \right\Vert^2
	= o_p(1)
	$,
	and we take $\sigma (g) =g$ by relabeling.
	Moreover, we observe that given $\beta$, the second term of $\hat Q_{N T} ( \beta , \gamma)$ does not affect the estimation of $\gamma$.
	Therefore, $\hat g_i ( \beta )$ defined in \eqref{eq-hatg-def}
	is also the estimate of $g_i$ given $\beta$ even if $\hat Q_{N T} ( \beta , \gamma)$ is our objective function.
	It follows that Lemma \ref{lem-g-consistent} applies for the GAGFL estimator.

	\subsubsection{Proof of Lemma \ref{lem-db-cb-h}}
		
	Let
	\begin{align*}
	\hat Q (\beta)
	=& \frac{1}{NT}\sum_{i=1}^N \sum_{t=1}^T
	(y_{it}  - x_{it} ' \beta_{ \hat g_i ( \beta) ,t} )^2
	+
	\lambda  \sum_{g\in \mathbb{G}}
	\sum_{t=2}^T \dot w_{g,t} \left\Vert \beta_{g, t} - \beta_{g, t-1}
	\right\Vert, \\
	\check Q (\beta)
	=& \frac{1}{NT}\sum_{i=1}^N \sum_{t=1}^T
	(y_{it}  - x_{it} ' \beta_{g_i^0 ,t} )^2, \text{ and }
	\dot Q (\beta)
	= \frac{1}{NT}\sum_{i=1}^N \sum_{t=1}^T
	(y_{it}  - x_{it} ' \beta_{ \hat g_i ( \beta) ,t} )^2.
	\end{align*}
	Note that
	$\hat Q(\beta) = \hat Q (\beta, \hat \gamma (\beta))$, $\check Q(\beta) = \dot Q (\beta, \gamma^0),$ and $\dot Q(\beta) = \dot Q (\beta, \hat \gamma ( \beta))$. Note also that 
	$\hat \beta = \arg \min_{\beta \in \mathcal{B}^{GT}} \hat Q (\beta)$, 
	$\check \beta = \arg \min_{\beta \in \mathcal{B}^{GT}} \check Q (\beta)$, and 
$\dot \beta = \arg \min_{\beta \in \mathcal{B}^{GT}} \dot Q (\beta).$

	\begin{proof}
		We first evaluate the difference between $\mathring Q (\beta)$ and $\hat Q (\beta)$.
		Note that
		$
		\mathring Q (\beta) - \hat Q (\beta)
		=  \check Q (\beta) - \dot Q (\beta).
		$
		Thus, the proof of Lemma \ref{lem-db-cb} implies that
		$
		\mathring Q ( \hat \beta) - \hat Q (\hat \beta) = o_p (T^{-\delta}).
		$
		Similarly, we have
		$
		\mathring Q ( \mathring \beta) - \hat Q (\mathring \beta) = o_p (T^{-\delta}).
		$

		Next, we evaluate the difference between $\mathring \beta$ and $\hat \beta$.
		By the definition of $\mathring \beta$ and $\hat \beta$, we have
		$
		0 \le \mathring Q(\hat \beta) - \mathring Q (\mathring \beta)
		= \hat Q(\hat \beta) - \hat Q (\mathring \beta) + o_p (T^{-\delta})
		\le o_p (T^{-\delta}).
		$
		Thus, we have
		\begin{align}
		\mathring Q(\hat \beta) - \mathring Q (\mathring \beta)
		= o_p (T^{-\delta}). \label{eq-cq-cb-rate-h}
		\end{align}
		We observe that
		\begin{align*}
		\mathring Q(\hat \beta) - \mathring Q (\mathring \beta)
		=& \frac{1}{NT}\sum_{i=1}^N \sum_{t=1}^T
		(x_{it} ' ( \mathring \beta_{g_i^0 ,t} - \hat \beta_{g_i^0 ,t})  )^2 +2\frac{1}{NT}\sum_{i=1}^N \sum_{t=1}^T
		(y_{it}  - x_{it} ' \mathring \beta_{g_i^0 ,t} )(x_{it} ' ( \mathring \beta_{g_i^0 ,t} - \hat \beta_{g_i^0 ,t})  )   \\
		&+ \lambda  \sum_{g\in \mathbb{G}}
		\sum_{t=2}^T \dot w_{g,t} \left\Vert \hat \beta_{g, t} - \hat \beta_{g, t-1}
		\right\Vert
		- \lambda  \sum_{g\in \mathbb{G}}
		\sum_{t=2}^T \dot w_{g,t} \left\Vert \mathring \beta_{g, t} - \mathring \beta_{g, t-1}
		\right\Vert.
		\end{align*}
		By the first order condition for $\mathring \beta_{g,t}$, we have
		\begin{align*}
		-2 \frac{1}{N T} \sum_{g_i^0 = g} (y_{it} - x_{it}' \mathring \beta_{g_i^0 ,t}  )x_{it}
		+ \lambda \dot w_{g,t} e_{g,t} -\lambda \dot w_{g, t+1} e_{g, t+1} =0,
		\end{align*}
		where $e_{g,1}= e_{g,T+1} =0$, for $2 \le t \le T$, $e_{g,t} = (\mathring \beta_{g, t} - \mathring \beta_{g, t-1} ) / \left\Vert \mathring \beta_{g, t} - \mathring \beta_{g, t-1}
		\right\Vert $ if $\mathring \beta_{g, t} - \mathring \beta_{g, t-1} \neq 0$ and
		$\| e_{g,t}\| \le 1$ otherwise.
		We thus have
		\begin{align*}
		& 2\frac{1}{NT}\sum_{i=1}^N \sum_{t=1}^T
		(y_{it}  - x_{it} ' \mathring \beta_{g_i^0 ,t} )(x_{it} ' ( \mathring \beta_{g_i^0 ,t} - \hat \beta_{g_i^0 ,t})  )   \\
		=& \lambda  \sum_{g\in \mathbb{G}}
		\sum_{t=1}^T ( \dot w_{g,t} e_{g,t} - \dot w_{g, t+1} e_{g, t+1})'( \mathring \beta_{g_i^0 ,t} - \hat \beta_{g_i^0 ,t})   \\
		=& \lambda  \sum_{g\in \mathbb{G}}
		\sum_{t=2}^T  \dot w_{g,t} e_{g,t}'( (\mathring \beta_{g_i^0 ,t} - \mathring \beta_{g_i^0 ,t-1})
		- ( \hat \beta_{g_i^0 ,t} - \hat \beta_{g_i^0 ,t-1})) .
		\end{align*}
		Let $\mathcal{T}_{m_g,g}$ be the set of $t$ such that $\mathring \beta_{g, t} - \mathring \beta_{g, t-1} \neq 0$ and $\mathcal{T}_{m_g,g}^c = \{ 2, \dots, T\} \backslash \mathcal{T}_{m_g,g}$.
		We have
		\begin{align*}
		&  \lambda  \sum_{g\in \mathbb{G}}
		\sum_{t=2}^T  \dot w_{g,t} e_{g,t}'( (\mathring \beta_{g_i^0 ,t} - \mathring \beta_{g_i^0 ,t-1})
		- ( \hat \beta_{g_i^0 ,t} - \hat \beta_{g_i^0 ,t-1})) \\
		&+ \lambda  \sum_{g\in \mathbb{G}}
		\sum_{t=2}^T \dot w_{g,t} \left\Vert \hat \beta_{g, t} - \hat \beta_{g, t-1}
		\right\Vert
		- \lambda  \sum_{g\in \mathbb{G}}
		\sum_{t=2}^T \dot w_{g,t} \left\Vert \mathring \beta_{g, t} - \mathring \beta_{g, t-1}
		\right\Vert \\
		=& \lambda  \sum_{g\in \mathbb{G}}
		\sum_{t \in \mathcal{T}_{m_g,g}^c} \dot w_{g,t}
		\left( \left\Vert \hat \beta_{g, t} - \hat \beta_{g, t-1}
		\right\Vert
		- e_{g,t}' (\hat \beta_{g, t} - \hat \beta_{g, t-1})
		\right) \\
		&+ \lambda  \sum_{g\in \mathbb{G}}
		\sum_{t \in \mathcal{T}_{m_g,g}} \dot w_{g,t}
		\left( \left\Vert \hat \beta_{g, t} - \hat \beta_{g, t-1}
		\right\Vert - \frac{ (\mathring \beta_{g, t} - \mathring \beta_{g, t-1})'
			( \hat \beta_{g, t} - \hat \beta_{g, t-1})
		}{ \left\Vert \mathring \beta_{g, t} - \mathring \beta_{g, t-1}
		\right\Vert  }  \right)
	\ge 0,
	\end{align*}
	where the last inequality follows by the CS inequality.
	This implies that
	\begin{align*}
	\mathring Q(\hat \beta) - \mathring Q (\mathring \beta)
	\ge   \frac{1}{NT}\sum_{i=1}^N \sum_{t=1}^T
	(x_{it} ' ( \mathring \beta_{g_i^0 ,t} - \hat \beta_{g_i^0 ,t})  )^2
	=& \frac{1}{T}\sum_{g \in \mathbb{G}} ( \mathring \beta_{g} - \hat \beta_{g})'
	M (\gamma^0, g, g) ( \mathring \beta_{g} - \hat \beta_{g}) \\
	\ge  & \hat \rho  \frac{1}{T} \sum_{g \in \mathbb{G}} \left\Vert \hat \beta_{g} - \mathring \beta_{g} \right\Vert^2.
	\end{align*}
	Hence, by \eqref{eq-cq-cb-rate-h} and Assumption \ref{a: iden}.\ref{as-m-eigen}, we have that,
	$
	1/T \sum_{g \in \mathbb{G}} \left\Vert \hat \beta_{g} - \mathring \beta_{g} \right\Vert^2 = o_p (T^{-\delta}),
	$
	which further implies that
	$
	\left\Vert \hat \beta_{g,t} - \mathring \beta_{g,t} \right\Vert^2 = o_p (T^{1-\delta})
	$
	for any $\delta$.
	This gives the desired result.
\end{proof}

\subsubsection{Proof of Theorem \ref{lem-break-c-agfl}}

\begin{proof}
	As $\mathring \beta$ minimizes $\hat Q (\beta, \gamma^0)$,
	$\mathring \beta = \hat \beta$ if $\hat \gamma = \gamma^0$.
	We note that
	\begin{align*}
	\Pr (\hat \gamma \neq \gamma^0)
	= \Pr \left( \max_{1 \le i \le N} \mathbf{1} \{ \hat g_i (\hat \beta)\neq g_i^0 \} = 1 \right)
	\le \sum_{i=1}^N E \left(\mathbf{1} \{\hat g_i (\hat \beta)\neq g_i^0 \}  \right).
	\end{align*}
	From Lemmas  \ref{lem-db-cb-h} and \ref{lem-lasso-c},
	we have $\Pr ( \hat \beta \in \mathcal{N}_{\eta}) \to 1$ for any $\eta$.
	Together with this,
	the argument made in the proof of Lemma \ref{lem-g-consistent}
	shows that
	$\max_{1 \le i \le N} E \left(\mathbf{1} \{ \hat g_i (\hat \beta)\neq g_i^0 \}  \right)= O (T^{-\delta})$ for any $\delta > 0$.
	This means that
	\begin{align*}
	\Pr (\hat \gamma \neq \gamma^0)
	\le N \max_{1 \le i \le N} E \left(\mathbf{1} \{ \hat g_i (\hat \beta)\neq g_i^0 \} \right)
	= o (NT^{-\delta})
	\end{align*}
	for any $\delta$.
	Thus, under the condition of the theorem,
	from Lemma  \ref{lem-break-c}, we have
	\begin{align*}
	& \Pr \left( \left\| \hat \theta_{g,t} \right\| \neq 0, \exists t \in \mathcal{T}_{m_g^0,g}^{0c}, g \in \mathbb{G}\right)  \\
	\le &
	\Pr \left( \left\{ \left\| \hat \theta_{g,t} \right\| \neq 0, \exists t \in \mathcal{T}_{m_g^0,g}^{0c}, g \in \mathbb{G} \right\} , \left\{ \hat \gamma = \gamma^0 \right\} \right)
	+ \Pr \left( \hat \gamma \neq \gamma^0 \right)  \\
	=& \Pr \left( \left\{ \left\| \mathring \theta_{g,t} \right\| \neq 0, \exists t \in \mathcal{T}_{m_g^0,g}^{0c}, g \in \mathbb{G} \right\} , \left\{ \hat \gamma = \gamma^0 \right\} \right)
	+ \Pr \left( \hat \gamma \neq \gamma^0\right) \\
	\le &
	\Pr \left( \left\| \mathring \theta_{g,t} \right\| \neq 0, \exists t \in \mathcal{T}_{m_g^0,g}^{0c}, g \in \mathbb{G}  \right)
	+ \Pr \left( \hat \gamma \neq \gamma^0\right)
	\to 0.
	\end{align*}
	We therefore have the desired result.
\end{proof}

\subsubsection{Proof of Theorem \ref{lem-agafl-date}}

\begin{proof}
	Given Lemma \ref{lem-db-cb-h} and Theorem \ref{lem-break-c-agfl}, the proof is based on an argument essentially identical to the proof of Corollary 3.4 in \citet{qian&su2016} and is thus omitted.
\end{proof}

\subsubsection{Proof of Theorem \ref{thm-mrb-ad-agfl}}

\begin{proof}
	The theorem holds using Lemmas \ref{lem-db-cb-h} and \ref{lem-mrb-ad}.
\end{proof}

\bibliography{reference}
\bibliographystyle{abbrvnat}

\end{document}